\documentclass[letterpaper,11pt]{article}
\usepackage[margin=1in]{geometry}

\usepackage{amsthm}
\usepackage{thmtools,thm-restate}
\usepackage{nicefrac}
\usepackage{booktabs} % For formal tables
\usepackage{graphicx}
\usepackage{amsmath}
\usepackage{enumitem}
\usepackage{amssymb}
\usepackage{amsfonts}
\usepackage{xcolor}
\usepackage{mathtools}
\usepackage{float}
\usepackage[ruled,vlined,linesnumbered]{algorithm2e} 
\usepackage{soul}
\setstcolor{blue}
\usepackage[normalem]{ulem}
\usepackage{tabu}
\usepackage[colorlinks=true,urlcolor=black,linkcolor=black,citecolor=black]{hyperref}

\newcommand{\MMS}{\text{MMS}}
\newcommand{\I}{\mathcal{I}}
\newcommand{\Rt}{R^{(t)}}
\newcommand{\Rk}{R^{(k)}}
\newcommand{\Tt}{T^{(t)}}
\newcommand{\Prn}{\{P_1, \dots, P_{t_3} \}}
\newcommand{\Tk}{T^{(k)}}
\newcommand{\ins}{\I=\langle N,M,V\rangle}
\newcommand{\MP}{\mathcal{P}}
\newcommand{\tf}[2]{\tfrac{#1}{#2}}
\newcommand{\FRAC}{\tf34+\gamma}

\newtheorem{theorem}{Theorem}[section]
\newtheorem{corollary}[theorem]{Corollary}
\newtheorem{lemma}[theorem]{Lemma}
\newtheorem{definition}[theorem]{Definition}
\newtheorem{example}[theorem]{Example}
\newtheorem{remark}[theorem]{Remark}
\newtheorem{observation}[theorem]{Observation}
\newtheorem{claim}[theorem]{Claim}

\SetAlFnt{\small}
\SetAlCapFnt{\small}
\SetAlCapNameFnt{\small}
\SetAlCapHSkip{0pt}
\IncMargin{-\parindent}

\title{An Improved Approximation Algorithm for Maximin Shares\thanks{Supported by NSF CAREER Award 1942321.}}
\author{Jugal Garg\thanks{University of Illinois at Urbana-Champaign.}\\\texttt{jugal@illinois.edu}
\and Setareh Taki\footnotemark[2]\\\texttt{staki2@illinois.edu}
}
\date{}
\begin{document}
\maketitle

\begin{abstract}
Fair division is a fundamental problem in various multi-agent settings, where the goal is to divide a set of resources among agents in a \emph{fair} manner. We study the case where $m$ indivisible items need to be divided among $n$ agents with additive valuations using the popular fairness notion of \emph{maximin share} (MMS). An MMS allocation provides each agent a bundle worth at least her maximin share. While it is known that such an allocation need not exist~\cite{procaccia2014fair,kurokawa2016can}, a series of remarkable work~\cite{procaccia2014fair,KurokawaPW18,amanatidis2017approximation,barman2017approximation, ourpaper} provided approximation algorithms for a $\tfrac{2}{3}$-MMS allocation in which each agent receives a bundle worth at least $\tfrac{2}{3}$ times her maximin share. More recently, Ghodsi et al.~\cite{ghodsi2017fair} showed the existence of a $\tf34$-MMS allocation and a PTAS to find a ($\tf34-\epsilon$)-MMS allocation for an $\epsilon > 0$. Most of the previous works utilize intricate algorithms and require agents' approximate MMS values, which are computationally expensive to obtain.  

In this paper, we develop a new approach that gives a simple algorithm for showing the existence of a $\tf34$-MMS allocation. Furthermore, our approach is powerful enough to be easily extended in two directions: First, we get a strongly polynomial time algorithm to find a $\tf34$-MMS allocation, where we do not need to approximate the MMS values at all. Second, we show that there always exists a $(\tf34+\tfrac{1}{12n})$-MMS allocation, improving the best previous factor. This improves the approximation guarantee, most notably for small $n$. We note that $\tf34$ was the best factor known for $n> 4$. 
\end{abstract}

\section{Introduction}
Fair division is a fundamental problem in various multi-agent settings, where the goal is to divide a set of resources among agents in a \emph{fair} manner. It has been a subject of intense study since the seminal work of Steinhaus~\cite{steinhaus1948problem} where he introduced the cake-cutting problem for $n>2$ agents: Given a heterogeneous (divisible) cake and a set of agents with different valuation functions, the problem is to find a \emph{fair} allocation. The two most well-studied notions of fairness are: 1) Envy-freeness, where each agent prefers her own share of cake over any other agents' share, and 2) Proportionality, where each agent receives a share that is worth at least $1/n$ of her value for the entire cake.

We study the \emph{discrete} fair division problem where $m$ \emph{indivisible} items need to be divided among $n$ agents with additive valuations. For this setting, no algorithm can provide either envy-freeness or proportionality, in general, e.g., consider allocating a single item among $n>1$ agents. This necessitated an alternate concept of fairness. Budish~\cite{budish2011combinatorial} introduced an intriguing option called \emph{maximin share}, which has attracted a lot of attention~\cite{procaccia2014fair, kurokawa2016can, ghodsi2017fair, ourpaper, KurokawaPW18, barman2017approximation, FarhadiGHLPSSY19, amanatidis2017approximation}. The idea is a straightforward generalization of the popular \emph{cut and choose protocol} in the cake-cutting problem and a  natural relaxation of proportionality. Suppose we ask an agent $i$ to partition the items into $n$ bundles (one for each agent), with the condition that the other $n-1$ agents get to choose a bundle before her. In the worst case, $i$ receives the least preferred bundle. Clearly, in such a situation, the agent will choose a partition that maximizes the value of her least preferred bundle. This maximum possible value is called $i$'s maximin share (MMS) value. In fact, when all agents have the same valuations, $i$ cannot guarantee more than the MMS value. 

Each agent's MMS value is a specific objective that gives her an intuitive measure of the fairness of an allocation. For example, Gates et al.~\cite{GatesGD20} showed that in real-life experiments maximin metric is preferred by participating agents over others. This raises a natural question: Is there an allocation where each agent receives a bundle worth at least her MMS value? An allocation satisfying this property is said to be maximin share allocation (MMS allocation), and if it exists, it provides strong fairness guarantees to each individual agent. However, Procaccia and Wang~\cite{procaccia2014fair}, through a clever counter-example, showed that MMS allocation might not exist but a $\tfrac{2}{3}$-MMS allocation always exists, i.e., an allocation where each agent receives a bundle worth at least $\tfrac{2}{3}$ of their MMS value. Later, Ghodsi et al.~\cite{ghodsi2017fair} improved the factor by showing the existence of a $\tf34$-MMS allocation using a sophisticated technique with a very challenging analysis.

We note that these are primarily existential results that do not provide any efficient algorithm to find such an allocation. The main issue in these techniques is the need for agents' MMS values. The problem of finding the MMS value of an agent is NP-hard,\footnote{Observe that the partition problem reduces to the MMS value problem with $n=2$.} but a polynomial-time approximation scheme (PTAS) exists~\cite{woeginger1997polynomial}. Theoretically, one can use PTAS to find a $(\tf34-\epsilon)$-MMS allocation for an $\epsilon>0$ in polynomial time. However, for practical purposes, such algorithms are not very useful for small $\epsilon$. Hence, finding an efficient algorithm to compute a $\tf34$-MMS allocation remained open. 

\subsection{Our Results and Techniques}
In this paper, we develop a new approach that gives a simple algorithm for showing the existence of a $\tf34$-MMS allocation. Furthermore, our approach is powerful enough to be easily extended in two directions: First, we get a strongly polynomial time algorithm to find a $\tf34$-MMS allocation, where we do not need to use the PTAS in \cite{woeginger1997polynomial} to approximate the MMS values at all. Second, we show that there always exists a $(\tf34+\tfrac{1}{12n})$-MMS allocation, improving the best previous factor by Ghodsi et al.~\cite{ghodsi2017fair}. This improves the approximation guarantee, most notably for small $n$. We note that there are works, e.g.,~\cite{amanatidis2017approximation,ghodsi2017fair}, exploring better approximation factors for a small number of agents, and $\tf34$ was the best factor known for $n> 4$. 

Our algorithms are extremely simple. We first describe the basic algorithm, given in Section~\ref{sec:MMS-known}, that shows the existence of a $\tf34$-MMS allocation. We assume that MMS values are known for all agents. Since the MMS problem is scale-invariant (shown in Lemma \ref{lem:Scale}), we scale valuations to make each agent's MMS value 1. Then, we assign \emph{high-value items} (e.g., a single item that some agent values at least $\tf34$) to agents, who value them at least $\tf34$, with a simple greedy approach based on the pigeonhole principle. We remove the assigned items and the agents receiving these items from further consideration. This reduces the number of high-value items to be at most $2n'$, where $n'$ is the number of remaining agents. These greedy assignments massively simplify allocation of high-value items, which was the most challenging part of previous algorithms. Next, we prepare $n'$ bags, one for each remaining agent, and put at most two high-value items in each bag. Then, we add low-value items on top of each of these bags one by one using a bag filling procedure until the value of bag for some agent is at least $\tf34$. The main technical challenge here is to show that there are enough low-value items to give every agent a bag they value at least $\tf34$. 

In Section \ref{sec:MMS-NOT-known}, we extend the basic algorithm to compute a $\tf34$-MMS allocation in strongly polynomial time without any need to compute the actual MMS values (using the PTAS in \cite{woeginger1997polynomial}). Here, we define a notion of \emph{tentative assignments} and a novel way for updating the MMS upper bound. For each agent, we use the average value, that is the value of all items divided by the number of agents, as an upper bound of her MMS value. The only change from the basic algorithm is that some of the greedy assignments are tentative, i.e., they are valid only if the current upper bound of the MMS values is tight enough. We show that this can be checked by using the total valuation of low-value items. If the upper bounds are not tight enough for some agents, then we update the MMS value of such an agent and repeat. We show that we do not need to update the MMS upper bounds more than $O(n^3)$ times before we have a good upper bound on all MMS values. Then, we show that the same bag filling procedure, as in the basic algorithm, satisfy every remaining agent. The running time of the entire algorithm is $O(nm(n^4 + \log{m}))$. 

In Section~\ref{sec:last}, we show that our basic algorithm also yields a better bound of the existence of a $(\tf34+\tfrac{1}{12n})$-MMS allocation. The entire algorithm remains exactly the same but with an involved analysis. The analysis is tricky in this case, so we add a set of \emph{dummy items} to make proofs easier. We use these items to make up for the extra loss for the remaining agents due to the additional factor, and, of course, these items are not assigned to any agent in the algorithm. 

\subsection{Related Work}
Maximin share is a popular fairness notion of allocating indivisible items among agents. Bouveret and Lema\^itre~\cite{bouveret2016characterizing} showed that an MMS allocation always exists in some restricted cases, e.g., when there are only two agents or if agents' valuations for items are either 0 or 1, but left the general case as an open problem. As mentioned earlier, Procaccia and Wang~\cite{procaccia2014fair} showed that MMS allocation might not exist, but a $\tfrac{2}{3}$-MMS allocation always exists. They also provided a polynomial time algorithm to find a $\tfrac{2}{3}$-MMS allocation when the number of agents $n$ is constant. For the special case of four agents, their algorithm finds a $\tf34$-MMS allocation. Amanatidis et al.~\cite{amanatidis2017approximation} improved this result by addressing the requirement for a constant number of agents, obtaining a PTAS that finds a $(\tfrac{2}{3}-\epsilon)$-MMS allocation for an arbitrary number of agents; see~\cite{KurokawaPW18} for an alternate proof. In~\cite{amanatidis2017approximation}, they also showed that a $\tfrac{7}{8}$ MMS allocation always exists when there are three agents. This factor was later improved to $\tfrac{8}{9}$ in~\cite{GourvesM19}. 

Taking a different approach, Barman and Krishnamurthy~\cite{barman2017approximation} obtained a greedy algorithm to find a $\tfrac{2}{3}$-MMS allocation. While their algorithm is fairly simple, the analysis is not. More recently, Garg et al.~\cite{ourpaper} obtained a simple algorithm to find a $\tfrac{2}{3}$-MMS allocation that also has a simple analysis.

Ghodsi et al.~\cite{ghodsi2017fair} improved these results by showing the existence of a $\tf34$-MMS allocation and a PTAS to find a $(\tf34-\epsilon)$ MMS allocation. 

Maximin share fairness has also been studied in many different setting, e.g., 
for asymmetric agents (i.e., agents with different entitlements)~\cite{FarhadiGHLPSSY19}, for group fairness~\cite{barman2018groupwise,chaudhury2020little}, beyond additive valuations~\cite{barman2017approximation, ghodsi2017fair,li2018fair}, in matroids \cite{GourvesM19}, with additional constraints~\cite{GourvesM19,biswas2018fair}, for agents with externalities~\cite{BranzeiMRLJ13,AhmadiPourAnariEGHIMM13}, with graph constraints~\cite{bei2019connected,lonc2019maximin}, for allocating chores~\cite{aziz2017algorithms, barman2017approximation, huang2019algorithmic}, and with strategic agents \cite{strategicagents,Amanatidis2016Truthful,Amanatidis2017Truthful,Aziz2019Strategyproof}.

\section{The MMS Problem and its Properties}
We consider the fair allocation of a set $M$ of $m$ indivisible items among a set $N$ of $n$ agents with additive valuations, using the popular notion of maximin share (MMS) as our measure of fairness. Let $v_{ij}$ denote agent $i$'s value for item $j$, and $i$'s valuation of any bundle $S\subseteq M$ of items is given by $v_i(S) = \sum_{j\in S} v_{ij}$. Let $V=(v_1, \dots, v_n)$ denote the set of all valuation functions.

An agent's \textit{MMS value} is defined as the maximum value she can guarantee herself if she is allowed to choose a partition of items into $n$ bundles (one for each agent), on the condition that other agents choose their bundles from the partition before her. In the following definition we define it formally.

\begin{definition}[MMS value and MMS partition]
\label{def:MMS-val}
Let $\ins$ denote an instance of the fair division problem, and let $\Pi_n(M) = \{P = \{P_1, \dots, P_n\}\ |\ P_i \cap P_j = \emptyset, \forall i, j;\ \cup_k P_k  = M\}$ be the set of all feasible partitions of $M$ into $n$ bundles (one for each agent).  Agent $i$'s MMS value or $\mu_i^n(M)$ (or simply $\mu_i$ when $n$ and $M$ are clear from the context) is defined as 
\begin{equation*}\label{eq:MMS}
    \mu_i^n(M) = \max_{P\in \Pi_n(M)} \min_{P_k \in P} v_i(P_k)\enspace .
\end{equation*}
We call a partition achieving $\mu_i$, an \textit{MMS partition} of agent $i$.

Further, let $\MP^n_i(M)$ denote the set of partitions achieving $\mu_i^n(M)$, i.e., 
\begin{equation*}\label{eq:MMSp}
   \MP^n_i(M) = \{  P\in \Pi_n(M): \min_{P_k \in P} v_i(P_k) = \mu_i^n(M)\}\enspace .
\end{equation*}
\end{definition}
In other words, $\MP^n_i(M)$ is set of all MMS partitions of agent $i$ for items in $M$ when there are $n$ agents.

\begin{definition}[$\alpha$-MMS allocation and MMS problem]
\label{def:MMS-alloc}
We say an allocation $A=(A_1, \dots, A_n)$ is $\alpha$-MMS, for $\alpha\in (0,1]$, if each agent $i$ receives a bundle $A_i$ worth at least $\alpha$ times her MMS value, i.e., $v_i(A_i) \ge \alpha\cdot \mu_i, \forall i\in N$. 
An MMS allocation is simply $1$-MMS allocation. 

Given an instance $\ins$ and an approximation factor $\alpha\in(0,1]$, the MMS problem is to find an $\alpha$-MMS allocation. 
\end{definition}

\subsection{Properties of Maximin Share}\label{sec:properties}
In this section, we state nice properties of maximin shares that our algorithm exploits. These are standard results appeared in~\cite{procaccia2014fair,amanatidis2017approximation,bouveret2016efficiency,barman2017approximation,ourpaper,ghodsi2017fair}. For completeness, we include their proofs in~\ref{sec:mp}. 

\begin{restatable}{lemma}{Average}\label{lem:avg}{\em \scshape (Average upper bounds MMS).}
$\mu_i^n(M) \leq \frac{v_i(M)}{n}, \forall i\in N$.
\end{restatable}

\begin{restatable}{lemma}{Scale}\label{lem:Scale}{\em \scshape (Scale Invariance).}
Let $A=(A_1, \dots, A_n)$ be an $\alpha$-MMS allocation for instance $\ins$. If we create an alternate instance $\I'=\langle N,M,V'\rangle$ where valuations of each agent $i$ are scaled by $c_i>0$, i.e., $v'_{ij}:= c_i\cdot v_{ij}, \forall j\in M$, then $\mu'_i = c_i\cdot \mu_i$ and $A$ is an $\alpha$-MMS allocation for $\I'$ .
\end{restatable}

\subsubsection{Ordered Instances}
We say that an instance $\ins$ is ordered if:
\[v_{i1} \ge v_{i2} \ge \cdots \ge v_{im}, \forall i\in N\text{.}\]
In words, in an ordered instance, all agents have the same order of preferences over items. Bouveret and Lema\^itre~\cite{bouveret2016efficiency} showed that the ordered instances are the worst case. They provided a reduction from any arbitrary instance $\ins$ to an ordered instance $\I'=\langle N,M,V'\rangle$ with a simple polynomial-time procedure that converts an MMS allocation of $\I'$ into an MMS allocation of $\I$. Later, Barman and Krishnamurthy~\cite{barman2017approximation} generalized this result for $\alpha$-MMS allocations. This property is used in~\cite{barman2017approximation, ourpaper} to find a $\tfrac{2}{3}$-MMS allocation. Observe that the $\MMS$ values of an agent in $\I$ and $\I'$ are the same because, by Definition \ref{def:MMS-val}, it neither depends on the order of the items nor on other agents' valuations. 

\begin{restatable}{lemma}{Order}\label{lem:order}{\em \scshape (Ordered Instance~\cite{bouveret2016efficiency,barman2017approximation}).}
Without loss of generality, we can assume that agents have the same order of preferences over the items, i.e., $v_{i1} \geq v_{i2} \geq \dots \geq v_{im},   \forall i \in N \text{.}$
\end{restatable}

\subsubsection{Bag Filling for Low Value Items}\label{sec:bagf}
Ghodsi et al.~\cite{ghodsi2017fair} showed that if we normalize the valuation of agents so that $\mu_i^{n}(M)=1, \forall i\in N$ and $v_{ij} \leq \beta, \forall i\in N, j\in M$, then we can find a $(1-\beta)$-MMS allocation using the following simple \emph{bag filling procedure}: All items are unallocated in the beginning. Start with an empty bag $B$ and keep filling it with unallocated items until some agent $i$ values $B$ at least $(1-\beta)$. Then, allocate $B$ to $i$ (choose $i$ arbitrarily if there are multiple such agents). Note that any remaining agent values $B$ at most one because before adding the last item to $B$ everyone values it less than $(1-\beta)$ and adding one item will not increase the value of $B$ more than $\beta$. We repeat this process for the set of unallocated items and the set of agents who have not allocated any bundle yet. Since $v_i(M) \ge |N|, \forall i$ using Lemma~\ref{lem:avg}, there are enough items to satisfy all the agents with a bag that they value $(1-\beta)$. 

In Sections~\ref{sec:Bag-Filling-Known} and \ref{sec:bag-filling-UNKNOWN}, we design a more general bag filling procedure.

\subsubsection{Reduction}\label{sec:red}
A useful concept of \emph{valid reduction} is used in~\cite{kurokawa2016can,KurokawaPW18,amanatidis2017approximation,ghodsi2017fair,ourpaper}. From Definition \ref{def:MMS-val}, $\mu^{k}_i(S)$ denote the MMS value of agent $i$ when $S$ is the set of items that needs to be divided among $k$ agents (including $i$). Recall that for the $\alpha$-MMS allocation problem for instance $\ins$, we want to partition $M$ into $|N|$ bundles $(A_1, \dots, A_{|N|})$ such that $v_i(A_i) \ge \alpha\cdot \mu^{|N|}_i(M), \forall i$. 

\begin{definition}[Valid reduction]\label{def:Valid-Reduction}
For obtaining an $\alpha$-MMS allocation, the act of removing a set $S\subseteq M$ of items and an agent $i$ from $M$ and $N$ is called a \emph{valid reduction} if
\begin{equation}\label{eqn:red}
\begin{aligned}
v_i(S) \ge \alpha\cdot \mu_i^{|N|}(M) \\
\mu_{i'}^{|N|-1}(M\setminus S) \ge \mu_{i'}^{|N|}(M), & \ \ \ \ \forall i'\in N\setminus \{i\}\enspace .
\end{aligned}
\end{equation}
\end{definition}

In words, valid reduction is the process of reducing the size of the instance $\ins$ by assigning a set of items $S$ to an agent $i$ and getting a smaller instance $\I'=\langle N \setminus \{i\}, M \setminus S, V \setminus \{v_i\} \rangle $ while the two conditions in \eqref{eqn:red} is satisfied. Clearly, an $\alpha$-MMS allocation for the smaller instance $\I'$ gives an $\alpha$-MMS allocation for the original instance. In our algorithms, we use it to remove \emph{high-value items} and get smaller instances.

\section{Existence of $\frac{3}{4}$-MMS Allocation} \label{sec:MMS-known}
In this section, we present a simple proof of the existence of a $\frac{3}{4}$-MMS allocation for a given instance $\ins$. We assume that the MMS value $\mu_i$ of each agent $i$ is given. Finding the exact $\mu_i$ is an NP-Hard problem, however a PTAS exists~\cite{woeginger1997polynomial}. This implies a PTAS to compute a $(\tf34-\epsilon)$-MMS allocation for any $\epsilon>0$. Using the properties stated in Section~\ref{sec:properties}, we normalize valuations so that $\mu_i^n(M)=1, \forall i$ (Lemma \ref{lem:Scale}) and assume that $\I$ is an ordered instance, i.e., $v_{i1} \geq \dots \geq v_{im}, \forall i$ (Lemma \ref{lem:order}). Our proof is algorithmic. Whenever we apply a valid reduction and more than one agent satisfies the conditions~\eqref{eqn:red}, we choose one arbitrarily.

For the ease of exposition, we abuse notation and use $M$ and $N$ to respectively denote the set of unallocated items and the set of agents who have not received any bundle yet. We also use $n:=|N|$ and $m:=|M|$. Further, we use $j$ to denote the $j^{\text{th}}$ highest value item in $M$.

\begin{algorithm}[t!]
\caption{$\alpha$-MMS Allocation} \label{algo:first}
\DontPrintSemicolon
  \SetKwFunction{Define}{Define}
  \SetKwInOut{Input}{Input}\SetKwInOut{Output}{Output}
  \Input{Ordered Instance $\ins$, i.e., $v_{i1} \ge v_{i2} \ge \dots \ge v_{im}, \forall i\in N$ and $\alpha$}
  \Output{$\alpha$-MMS Allocation}
  \BlankLine
  Normalize Valuations \tcp*{Scale valuations so that $\mu_i^n(M)=1, \forall i$}
  $(N,M,V) \gets$ Initial-Assignment$(N,M,V,\alpha) \label{line:init} $\tcp*{Algorithm~\ref{algo:initial}}
 
  Bag-Filling$(N,M,V,\alpha)$ \tcp*{Algorithm~\ref{algo:Bag-Filling-MAIN}}
\end{algorithm}

The algorithm is given in Algorithm~\ref{algo:first}. We use $\alpha=\frac{3}{4}$ in this section. Later, in Section~\ref{sec:last}, we use the same algorithm for $\alpha=\tf34 + \tfrac{1}{12n}$. Algorithm~\ref{algo:first} has two main parts: Initial Assignment and Bag Filling. We now describe each part separately in detail.

\subsection{Initial Assignment}\label{sec:initial-assignment}
We first assign \emph{high-value items} using Algorithm~\ref{algo:initial}. We note that handling high-value items is the biggest challenge in the MMS problem, e.g., a major part of Ghodsi et al. algorithm~\cite{ghodsi2017fair} is devoted to allocating high-value items. In our algorithms, we make it a simple process of greedy assignment by leveraging the pigeonhole principle to make valid reductions. We define bundles $S_1:=\{1\}$, $S_2:=\{n, n+1\}$ ($\emptyset$, if $m \le n$), $S_3:=\{2n-1, 2n, 2n+1\}$
($\emptyset$, if $m \le 2n$), and $S_4:=\{1, 2n+1\}$ ($\emptyset$, if $m \le 2n$), where $S_1$ has the highest value item in $M$, $S_2$ has the $n^{\text{th}}$ and $(n+1)^{\text{th}}$ highest valued items in $M$, and so on. We will show that allocating any of these bundles to an agent who values it at least $\tf34$ is a valid reduction. Note that these bundles change after every valid reduction.  

\begin{algorithm}[t!]
\caption{Initial-Assignment}
\label{algo:initial}
\DontPrintSemicolon
  \SetKwFunction{Define}{Define}
  \SetKwInOut{Input}{Input}\SetKwInOut{Output}{Output}
  \Input{Ordered Instance $\ins$, where $\mu_i^n(M) = 1, \forall i\in N$, and an approximation factor $\alpha$}
  \Output{Reduced Instance}
  \BlankLine
  For any $S\subseteq M$, define $\Gamma(S):= \{i \in N: v_i(S) \geq \alpha \}$\;
  $n \gets |N|$\tcp*{$n$ changes with $N$}
  $S_1 \gets \{1\};\  S_2 \gets \{n, n+1\};\ S_3 \gets \{2n-1, 2n, 2n+1\};\ S_4 \gets \{1, 2n+1\}$\tcp*{bundles that can be assigned}
   \BlankLine
     \While{$(\Gamma(S_1)\cup \Gamma(S_2) \cup \Gamma(S_3) \cup \Gamma(S_4)) \neq \emptyset$ \label{line:while}}{ 
     $S\gets$ the lowest index bundle in $\{S_1, S_2, S_3, S_4\}$ for which $ \Gamma(S) \neq \emptyset$\;\label{line:low}
     $i\gets$ an agent in $\Gamma(S)$\;
     Assign $S$ to agent $i$\tcp*{initial assignment}
     $M\gets M\setminus S;\ \  N\gets N\setminus\{i\}$\label{line:end}\;
     }
  \Return{$\langle N, M, V\rangle$}\;
\end{algorithm}

In Algorithm \ref{algo:initial}, we keep assigning the lowest index $S\in \{S_1,S_2,S_3,S_4\}$ to agent $i$, if any, for which $v_i(S)\geq \frac{3}{4}$. Then, we update $M$ and $N$ to reflect the current unallocated items and agents who have not been assigned with any bundle yet. 
The following lemma extends the ideas that appeared in~\cite{ghodsi2017fair,ourpaper}.

\begin{lemma}\label{lem:valid-reduction}
Let $S$ be the lowest index bundle in $S\in\{S_1,S_2,S_3,S_4\}$ for which $\Gamma(S):= \{i \in N: v_i(S) \geq \tf34 \}$ is non-empty. Then, removing $S$ and agent $i$ with $v_i(S) \ge \tf34$ is a valid reduction.
\end{lemma}

\begin{proof} Clearly, $v_i(S) \geq \frac{3}{4}$. Therefore, we only need to show the second condition in Definition~\ref{def:Valid-Reduction}. We show this separately for each case of $S\in\{S_1,S_2, S_3,S_4\}$. Fix agent $i' \in N \setminus \{i\}$ and $P\in \MP^n_{i'}(M)$ (Recall from Definition \ref{def:MMS-val} that $\MP^n_{i'}(M)$ denote the set of partitions achieving $\mu^n_{i'}(M)$). We show that after removing $S$, there exists a partition of $M \setminus S$ into $(n-1)$ bundles such that the value of each bundle is at least $\mu_{i'}^n(M)$, i.e., $\mu^{n-1}_{i'}(M\setminus S) \ge \mu_{i'}^{n}(M)$. 
\begin{itemize}
    \item \textbf{$S=S_1$. } Removal of one item from $P$ affects exactly one bundle and each of the remaining $(n-1)$ bundles has value at least $\mu_{i'}^{n}(M)$. Therefore, $\mu^{n-1}_{i'}(M\setminus \{1\}) \ge \mu_{i'}^{n}(M), \forall i'\in N \setminus \{i \}$.
    \item \textbf{$S=S_2$. } In $P$, there exists a bundle with two items from $\{1, \dots, n+1\}$ (pigeonhole principle). Let $T$ be a bundle in $P$ that has two items from $\{1, \dots,n+1\}$. Let us exchange these items with items $n$ and $n+1$ in other bundles and distribute any remaining items in $T$ among other bundles arbitrarily. Clearly, the value of other bundles except $T$ does not decrease, and hence $\mu^{n-1}_{i'}(M\setminus \{n,n+1\}) \ge \mu_{i'}^{n}(M), \forall i'\in N\setminus \{i \}$.
    \item \textbf{$S=S_3$. } Similar to the proof of Case $`S = S_2$'. 
    \item \textbf{$S=S_4$.} In each iteration, the lowest index bundle from $\{S_1,S_2,S_3,S_4 \}$ is picked. Therefore, $S_4$ is only picked when $v_i(S_1),v_i(S_3) < \tf34$ for all $i \in N$ which implies that $v_{i'1} < \tf34$ and $v_{i'(2n+1)} < \tfrac{1}{4}$ and hence $v_i(S_4) < 1$ for all $i\in N$.
    
    In $P$, if items $1$ and $2n+1$ are in the same bundle, then clearly, removing $S_4$ and agent $i$ is a valid reduction. For the other case, if $1$ and $2n+1$ are in two different bundles, then we can make two new bundles, one with $\{1,2n+1\}$ and another with all the remaining items of the two bundles. The value of the bundle without $\{1,2n+1\}$ is at least $2\mu_i -1 \ge  \mu_i$ because $v_i(S_4) < 1$ for all $i\in N$ and $\mu_i \ge 1$. 
   Hence, this is a valid reduction. \qedhere
\end{itemize}
\end{proof}

Lemma~\ref{lem:valid-reduction} implies the following corollary. 

\begin{corollary}\label{cor:mu>1}
After the execution of Algorithm \ref{algo:initial} is completed, $\mu_i \ge 1$ for all $i \in N$.
\end{corollary}

Note that the original MMS value is 1 for each agent. Although it may increase after a valid reduction, 
we only need to give each agent a bundle of value at least $\tf34$ to achieve a $\tf34$-MMS allocation for the original instance. For this reason, we are interested in partitions where the value of each bundle is at least 1 for an agent. Therefore, we abuse notation to denote by $\MP^n_i(M)$ the set of partitions achieving the original MMS value, i.e., 
\begin{equation}\label{eq:MMSp1}
 \MP^n_i(M) = \{  P\in \Pi_n(M): \min_{P_k \in P} v_i(P_k) \ge 1\}\enspace .
\end{equation}
Observe that $\MP^n_i(M)$ contains MMS partitions of agent $i$ for the current $M$ and $n$ as well as the partitions that do not achieve the current $\mu_i$, but the value of each bundle is at least 1. 
\subsection{Bag Filling}\label{sec:Bag-Filling-Known}
\begin{algorithm}[t!]
\caption{Bag Filling Algorithm for $\alpha$-MMS Allocation}\label{algo:Bag-Filling-MAIN}
  \SetKwFunction{Define}{Define}
  \SetKwInOut{Input}{Input}\SetKwInOut{Output}{Output}
  \Input{Reduced Instance $\ins$ with $\mu^n_i =1 \ \forall i$ and an approximation factor $\alpha$}
  \Output{Allocation $A=(A_1, \dots, A_n)$ where $v_i(A_i) \geq \alpha$}
  \BlankLine
  \DontPrintSemicolon
  Initialize Bags $B=\{B_k\}_{k\in [n]}$ as in~\eqref{eq:B_i}\tcp*{see Figure~\ref{fig:bags}}
  $R \gets M \setminus J$\;
  \For{$k$ = 1 to $n$}{
  $T \gets B_k$ \label{line:four}\;
  Define $\Gamma(T) := \{i\in N\ : \ v_i(T) \geq \alpha\} $ \label{line:already} \;
  \While{$\Gamma(T)=\emptyset$\label{line:less}} {
  	$j \in R$ \tcp*{pick one low value item arbitrarily}
  	$T \gets T\cup \{j\}$;\ \ $R \gets R \setminus \{j\}$\tcp*{add the new item to bag}}
  	$ i\in \Gamma(T);\ A_i \gets T$;\ \ $N\gets N\setminus \{i\}$\tcp*{assign $T$ to $i\in \Gamma(T)$}}
\end{algorithm}

We use the bag filling procedure given in Algorithm~\ref{algo:Bag-Filling-MAIN} to satisfy the remaining agents. Let $J_1 := \{1,\dots,n\}$ denote the set of first $n$ items. Similarly, let us define $J_2:=\{n+1,\dots,2n\}$ and $J:=J_1 \cup J_2$. We call $J$ to be the set of high-value items. The following corollary is straightforward.

\begin{corollary}\label{cor:Bound-on-Values}
If $v_i(S)< \frac{3}{4}$, for all $i$ and for all $S \in\{S_1,S_2,S_3\}$, then $(i)$ $v_{ij} < \frac{3}{4},\ \forall j \in J_1$,\ $(ii)$ $v_{ij} < \frac{3}{8},\ \forall j \in J_2$, and $v_{in} <\frac{3}{4} - v_{i(n+1)}$, and $(iii)$ $v_{ij} < \frac{1}{4},\ \forall j \in M\setminus J$, for all $i$.
\end{corollary}

Next, we initialize $n$ bags as follows:  
\begin{equation}\label{eq:B_i}
B=\{B_1,B_2,\dots,B_n\}, \text{ where } \ B_k = \{k, 2n-k+1\}, \forall k \enspace .
 \end{equation}
 
Each bag contains one item from $J_1$ and one item from $J_2$ such that from $B_1$ to $B_n$ value of items from $J_1$ decreases and value of items from $J_2$ increases (see Figure \ref{fig:bags} for an illustration). 

 \begin{figure}[h]
     \centering
     \includegraphics[width=0.8\textwidth]{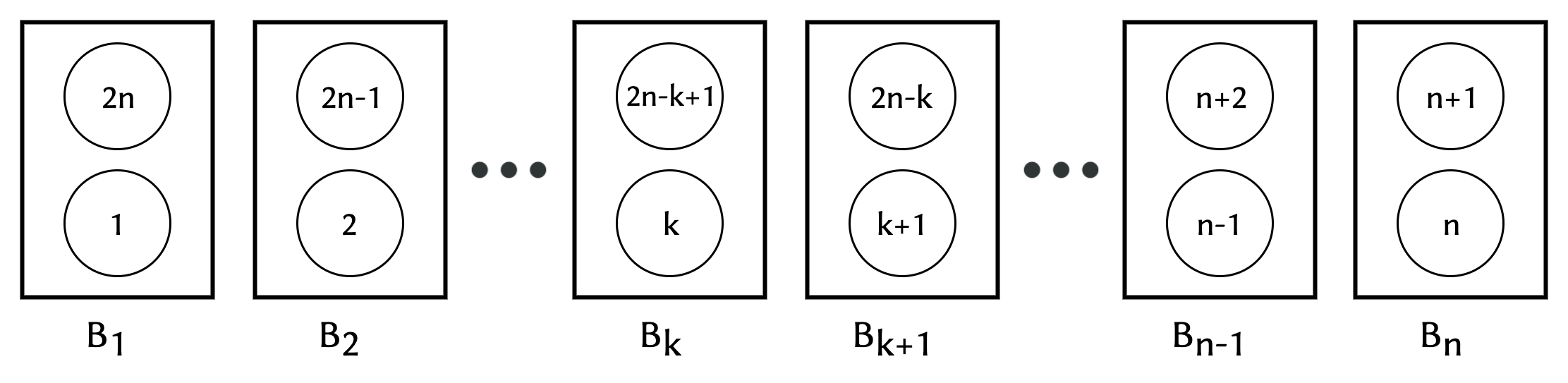}
     \caption{Setting of the items in the bags}
     \label{fig:bags}
 \end{figure}
 
Algorithm~\ref{algo:Bag-Filling-MAIN} has $n$ rounds. In each round $k$, it starts a new bundle $T$ with $T \gets B_k$. If there is an agent who values $T$ to be at least $\tf34$, then assign $T$ to such an agent. Otherwise, keep adding items from $M\setminus J$ to $T$ one by one until an agent with no bundle assigned to her values $T$ at least $\tf34$. The algorithm allocates $T$ to that agent, and if there are multiple such agents, it chooses one arbitrarily. 

For correctness, we need to show that there are enough items in $M\setminus J$ to add on top of each bag in~\eqref{eq:B_i} so that each agent gets a bundle that they value at least $\tf34$. For this, we first divide agents into two types: 
\[N^1:=\{i\in N\ |\ v_i(B_k) \le 1, \forall k\} \text{ and } N^2 := N\setminus N^1\enspace .\] 

If $N^2$ is empty, then it is easy to check that using Corollary \ref{cor:Bound-on-Values}(iii) and the ideas in Section~\ref{sec:bagf} that Algorithm~\ref{algo:Bag-Filling-MAIN} gives each agent at least $\tf34$. We need some more notation to show correctness when $N^2$ is not empty.

For an agent $i\in N^2$, define
\begin{equation}\label{eqn:LKx1}
\begin{aligned}
    L_i := & \ \{B_k: v_i(B_k) < \tf34 \}; \ \ \ \ l_i:=|L_i| \\ 
    H_i := & \ \{B_k: v_i(B_k) > 1\}; \ \ \ \ h_i := |H_i|\\
    x_i := & \ (\tf34)l_i - \sum_{k:B_k \in L_i} v_i(B_k). 
\end{aligned}
\end{equation}
In words, $L_i$ is the set of bags that $i$ values strictly less than $\tf34$, $H_i$ is the set of bags that $i$ values strictly more than $1$, and $l_i$ and $h_i$ are the number of bags in $L_i$ and $H_i$ respectively. Further, $x_i$ is the least total value needed to make all bundles in $L_i$ at least $\tf34$. Agents in $N^2$ have nice properties that we show in the following lemma.

\begin{lemma}\label{lem:three}
For an agent $i\in N^2$, $(i)$ $l_i >0$ and $h_i > 0$, $(ii)$ $v_{i1} > \tfrac{5}{8}$, $(iii)$ $v_i(B_k) <\tfrac{9}{8}, \forall k$, and $(iv)$ $v_{ij} < \tf18, \forall j \in M\setminus J$.
\end{lemma}

\begin{proof}
For the first part, if $h_i = 0$ then $i\in N^1$. Further, $l_i= 0$ implies that $v_i(B_n)\geq \tf34$, which cannot be true after the execution of Algorithm \ref{algo:initial} due to Line \ref{line:while}--\ref{line:end} of the algorithm.

For the second part, $v_{ij} < \tfrac{3}{8}, \forall j\in J_2$ due to Corollary \ref{cor:Bound-on-Values}. Since $h_i>0$, there exists a $j\in J_1$ such that $v_{ij} > \tfrac{5}{8}$. Further, since item $1$ is the highest value item for every agent, $v_{i1}>\tfrac{5}{8}$.

For the third part, each item of $J_1$ has value less than $\tf34$ and each item of $J_2$ has value less than $\tfrac{3}{8}$ (Corollary \ref{cor:Bound-on-Values}) for any agent. Therefore, $v_i(B_k) <\tfrac{9}{8}$ for any bundle $B_k$. 

For the fourth part, the value of item 1 for an agent $i\in N^2$ is more than $\tfrac{5}{8}$ and $v_i(S_4) < \tf34$, hence $v_{ij} \leq v_{i(2n+1)} < \tf18, \forall j\in M\setminus J$. 
\end{proof}
 
In the following lemma, we show that if the total value of the items in $M \setminus J$ for each agent $i$ in $N^2$ is at least $x_i + \tfrac{l_i}{8} - \tf18$, i.e., 
\[v_i(M\setminus J) \ge x_i + \tfrac{l_i}{8} -\tf18\enspace,\] 
then the bag filling algorithm will assign every agent (in $N^1$ and $N^2$) a bundle with value at least $\tf34$. In the rest of this section, we show that the bound on the value of $M \setminus J$ actually holds by using the fact that $\mu_i \ge 1$ for all $i \in N$.

\begin{lemma}\label{lem:bag-filling-terminates}
If $v_i(M\setminus J) \ge x_i + \tfrac{l_i}{8} -\tf18, \forall i\in N^2$, then Algorithm~\ref{algo:Bag-Filling-MAIN} gives every agent a bundle that they value at least $\tf34$.
\end{lemma}
\begin{proof}
This is proof by contradiction. Note that, in Algorithm \ref{algo:Bag-Filling-MAIN}, $R$ is the set of unallocated items from $M \setminus J$ and $T$ is the bag that is being filled at a time. Let $\Rk$ and $\Tk$ be respectively $R$ in the beginning and $T$ at the end of round $k$ of the algorithm, i.e., $R^{(1)}=M \setminus J$ and $T^{(1)} \supseteq B_1$.  For contradiction, suppose the algorithm stops at round $t$ because there are not enough unallocated items in $\Rt$ to satisfy any remaining agent $i$, i.e., $v_i(B_t \cup \Rt)< \frac{3}{4}$. 

If $i\in N^1$, each removed bundle in rounds $k \in [t-1]$, has value of at most $1$ for agent $i$. Because, if $v_i(B_k) \ge \tf34$ for $k \in [t-1]$, agent $i$ is already interested in $B_k$ ($i \in \Gamma(B_k)$) and the algorithm does not enter the while loop in Line \ref{line:less}. Therefore, if $v_i(B_k) \ge \tf34$ no more item has been added to $\Tk= B_k$. Also, if $v_i(B_k) < \tf34$ for $k \in [t-1]$, before adding the last item (if any) to $\Tk$, the value of $\Tk$ is less than $\tf34$ (otherwise, it would be out of loop and allocated to someone). Moreover, from Corollary \ref{cor:Bound-on-Values}, $v_{ij} < \tfrac{1}{4}$ for $j \in \Rk$ (since $\Rk \subseteq M \setminus J$). Therefore, at the end of the round $k$, $v_i(\Tk) < 1$. Further, since $v_i(M)\ge n$ and $v_i(B_k) \le 1, \forall k$, we have $v_i(B_t \cup \Rt) \ge 1$, which is a contradiction. 

If $i\in N^2$, then since at round $t$, $v_i(\Tt) < \tf34$, we have $B_t \in L_i$ (See \eqref{eqn:LKx1} for the definition of $L_i$). Consider a round $k \in [t-1]$. If $B_k\not\in L_i$, then $\Tk= B_k$ has been assigned to someone with no additional items added to $\Tk$ from $\Rk$ because $i \in \Gamma(T)$ and the algorithm does not enter the while loop in Line~\ref{line:less}. If $B_k \in L_i$, then in round $k$, before adding the last item (if any) to $\Tk$, the value of $i$ for $\Tk$ is less that $\tf34$. Moreover, from Lemma \ref{lem:three}, each item in $\Rk$ has value of at most $\tf18$. Therefore, if $B_k \in L_i$, the value of the assigned bag for $i$ in round $k$ is less than $\tfrac{7}{8}$. Since $B_t\in L_i$, at most $l_i-1$ bags from $L_i$ have been assigned up to $t-1$ iterations. Further, since items from $M\setminus J$ are added to bags in $L_i$ only, the total value taken from $M\setminus J$ up to $t-1$ iterations, according to agent $i$, is at most $x_i - (\tf34-B_t) + (l_i-1)/8$ where $x_i - (3/4-B_t)$ to make each of $L_i\setminus B_t$ exactly $\tf34$ (See \eqref{eqn:LKx1} for the definition of $x_i$) and $(l_i-1)/8$ to add an extra $\tf18$ to each. 
Hence, in the beginning of round $t$, 
\begin{equation}\label{eq:bN2}
    v_i(\Rt) \ge \left(x_i+ \tfrac{l_i}{8} - \tf18 \right) - \left(x_i- (\tf34 - v_i(B_t)) +(l_i-1)/8 \right) = \tf34 - v_i(B_t) \ ,
\end{equation} 
which is a contradiction. 
\end{proof}

Now, we only need to show that for each $i\in N^2$, we have
\begin{equation}\label{eqn:mb1}
     v_i(M\setminus J) \ge x_i + \tfrac{l_i}{8} -\tf18 \enspace .
\end{equation} 

\subsubsection{Showing~\eqref{eqn:mb1}}
In fact, we will show a stronger bound without $-\tf18$ in Theorem~\ref{thm:fb}. We will use the extra $-\tf18$ to improve the bound in Section~\ref{sec:last}. We start with a few lemmas to show more properties of agents in $N^2$, in addition to previous properties shown in Lemma \ref{lem:three}. Recall from \eqref{eq:MMSp1} that $\MP^n_{i}(M)$ denote the set of partitions where the value of each bundle is at least $1$. 

\begin{lemma}\label{lem:Giver-Bundle}
For an agent $i\in N^2$, there exists a bundle $P_k$ in every partition $P =\{P_1, \dots, P_n\} \in \MP^{n}_{i}(M)$ such that $v_i(P_k \setminus J) >\tfrac{1}{4}$.
\end{lemma}

\begin{proof}
If there exists a bundle $P_k$ with exactly one item $j$ from $J$, then $v_i(P_k \setminus J) \ge 1 - v_{ij}$ $> \tfrac{1}{4}$ because the value of every item is less than $\tf34$. Otherwise, each bundle has exactly two items from $J$, which implies that one of the bundles, say $P_{k}$, has two items $j_1,j_2$ from the set $\{n\} \cup J_2$. Since $v_{ij_1} + v_{ij_2} \leq v_{in} + v_{i(n+1)} <\tf34 $, $v_i(P_{k}\setminus J) >\tfrac{1}{4}$.
\end{proof}

Next, we show that there exists a partition in $\MP^n_i(M)$ for $i \in N^2$ where all items with value more than $\tfrac{5}{8}$ are in separate bundles. The intuition of this proof is that the bundle which has two items of value greater than $\tfrac{5}{8}$ can be merged with another bundle (or possibly two other bundles) and make two (or three) bundles each value at least 1. This basically utilizes the extra $\tfrac{2}{8}$ value in the first bundle to reshuffle items to obtain the desired partition. We begin with the following claim.

\begin{claim}\label{clm:1}
If there exists a partition $P = \{P_1, \dots, P_n\} \in \MP^{n}_{i}(M)$ for an agent $i\in N^2$ where a bundle $P_k\in P$ contains two items with value more than $\tfrac{5}{8}$ for agent $i$. Then, 
\begin{enumerate}
\item there exists another bundle $P_{k'} \in P$ for which $\max_{j \in P_{k'}} v_{ij} < \tfrac{3}{8}$. 
\item if $v_i(P_k \setminus J) > \tfrac{1}{4} $, then we can make two new bundles from the items in $P_k \cup P_{k'}$ where each bundle has one item with value more than $\tfrac{5}{8}$ and each bundle values at least 1.
\item if $P_{k'}$ cannot be divided into two parts with value at least $\tfrac{3}{8}$ each. Then, $v_i(P_{k'} \setminus J) < \tfrac{1}{4}$. 
\end{enumerate}
\end{claim}
\begin{proof}
Corollary~\ref{cor:Bound-on-Values} implies that the only items that can have value more than $\tfrac{5}{8}$ are items in $J_1 = \{1,\dots,n \}$. Therefore, $P_k$ has two items from $J_1$. By the pigeonhole principle, there must exist another bundle $P_{k'}$ where all items are from $\{n+1,\dots,m \}$. Corollary~\ref{cor:Bound-on-Values} implies that each item in $P_{k'}$ has value less than $\tfrac{3}{8}$. This proves the first part. 

For the second part, if $v_i(P_k \setminus J) > \tfrac{1}{4} $ we have $v_i(P_k) > \tfrac{5}{8} + \tfrac{5}{8} + \tfrac{1}{4} = \tfrac{3}{2}$. We make two bundles by initializing two empty bags and adding items from $P_k \cup P_{k'}$ one by one in decreasing order of value to the bundle which has the lower value. We will get new bundles with values $v_1$ and $v_2$ where $v_1 + v_2 > \tfrac{5}{2}$ and $|v_1 - v_2| < \tfrac{3}{8}$. The second inequality holds because each item in $P_{k'}$ and $P_k \setminus J$ has value less than $\tfrac{3}{8}$. This proves the second part. 

For the third part, assume $v_i(P_{k’} \setminus J) \ge \tfrac{1}{4}$ for a contradiction. We initialize two empty bundles and add items from $P_{k’}$ one by one in decreasing order of value to the bundle with lower value. We get two new bundles with value $v_1$ and $v_2$ where $v_1 + v_2 = v_i(P_{k’}) \ge 1 $ and $|v_1 - v_2| < \tf18$, which implies $v_1,v_2 \ge \tfrac{3}{8}$. The second inequality holds because each item in $P_{k’}\setminus J$ has value less than $\tf18$.
\end{proof}

\begin{lemma}\label{lem:sp}
For every $i$ $\in N^2$, there exists a partition $P = \{P_1, \dots, P_n\} \in \MP^{n}_{i}(M)$ such that each $P_k\in P$ has at most one item $j$ with $v_{ij} > \tfrac{5}{8}$. 
\end{lemma}

\begin{proof}
If there are more than two items with value more than $\tfrac{5}{8}$ in a bundle $P_k$ of $P\in \MP^{n}_{i}(M)$, then we add one of these items to $P_{k'}$, defined in Claim \ref{clm:1}. This will ensure that the value of both $P_k$ and $P_{k'}$ is at least $1$. By repeating this, we can obtain a $P\in \MP^{n}_{i}(M)$ that has at most two items with value more than $\tfrac{5}{8}$.

Next, we show that if there are two items $j_1, j_2$ each with value more than $\tfrac{5}{8}$ for an agent $i$ in a bundle $P_k$ of $P \in \MP^{n}_{i}(M)$, then we can construct another $P'\in \MP^n_{i}(M)$ where this is not true. Let $P_{k'} \in P$ be a bundle for which $\max_{j \in P_{k'}} v_{ij} < \tfrac{3}{8}$ (see Claim {\ref{clm:1}}(1) for the proof of its existence).

\textbf{Case 1:} If $v_i(P_k \setminus J) > \tfrac{1}{4}$, using Claim \ref{clm:1}(2), we make two bundles with value at least 1 and exactly one item with value more than $\tfrac{5}{8}$ in each.

\textbf{Case 2:} If $v_i(P_k \setminus J) \le \tfrac{1}{4}$ and there exists a partition $Q_{k'}^1$ and $Q_{k'}^2$ of items in $P_{k'}$ such that value of each $Q_{k'}^1$ and $Q_{k'}^2$ is at least $\tfrac{3}{8}$, then we can rearrange items in $P_k \cup P_{k'}$ and make two new bundles $(\{j_1\} \cup Q_{k'}^1)$ and $(\{j_2\} \cup Q_{k'}^2)$. Clearly, the value of each bundle is at least $1$ and each has exactly one item with value more than $\tfrac{5}{8}$.

\textbf{Case 3:} Finally, if $v_i(P_k \setminus J) \le \tfrac{1}{4}$ and no such $Q_{k'}^1$ and $Q_{k'}^2$ exists (as in Case 2), then we claim that there exists a partition of $P_{k'}$ into three bundles, each with value less than $\tfrac{3}{8}$. We can find this partition as follows: Initialize three empty bundles, and repeatedly add the highest value item of $P_{k'}$ to the bundle with the lowest value. For a contradiction, suppose one of the three bundles has value more than $\tfrac{3}{8}$, then the sum of the values of the other two sets must be less than $\tfrac{3}{8}$ because otherwise, they make a partition of two where each has value more than $\tfrac{3}{8}$. This means that at least one of the bundles must have a value less than $\tfrac{3}{16}$. This implies that the value of the largest bag before adding the last item must be less than $\tfrac{3}{16}$ and the last item also should value less than $\tfrac{3}{16}$, which is a contradiction. Therefore, there exists a partition $Q_{k'}^1$, $Q_{k'}^2$ and $Q_{k'}^3$ of $P_{k'}$ such that each has value less than $\tfrac{3}{8}$.

According to Lemma \ref{lem:Giver-Bundle}, there exists a bundle $P_{\hat{k}} \in \MP^{n}_{i}(M)$ such that $v_i(P_{\hat{k}} \setminus J) > \tfrac{1}{4}$. Let $v_i(P_{\hat{k}})=1+\delta$ for some $\delta \geq 0$. $P_{\hat{k}}$ cannot be same as $P_{k}$ (Case 1 of this proof) and $P_{k'}$ (Claim \ref{clm:1}(3)). We initialize three bags: $$\text{bag 1}: \{j_1\} \cup Q_{k'}^1 \hspace{0.8cm} \text{bag 2}: \{j_2\} \cup Q_{k'}^2 \hspace{0.8cm} \text{bag 3}:(P_k \setminus \{j_1,j_2\}) \cup (P_{\hat{k}} \cap J) \cup Q_{k'}^3 $$

Observe that the value of each of bag 1 and bag 2 is at most $\tfrac{9}{8}$, 
and the total value of all items in $P_k$, $P_{k'}$ and $P_{\hat{k}}$ is at least $3+\tfrac{2}{8}+\delta$. We sort the remaining items in decreasing order and add them one by one to a bag with the lowest value. Since the value of the last item added is at most $\tf18$ (Lemma \ref{lem:three}), each bag has a value of at least 1 and it has at most one item with value more than $\tfrac{5}{8}$. We repeat this process for each bundle with two items of value more than $\tfrac{5}{8}$ to find a desired partition. 
\end{proof}

Let $a$ be an agent in $N^2$. For simplicity, until the end of this section, when we use \textit{value} of an item or a bundle, we mean the value for agent $a$ (unless mentioned otherwise). Recall that we need to show \eqref{eqn:mb1}. Let $P=(P_1, \dots, P_n)\in \MP^n_a(M)$ be a partition satisfying the condition of Lemma~\ref{lem:sp}, i.e., there is at most one item with value greater than $\tfrac{5}{8}$ in $P_k$ for all $k \in [n]$. 

Using \eqref{eq:B_i}, we manipulate $P$ as follows: First, for each $B_i= \{j,j' \}$ with $\tf34 \leq v_a(B_i) \leq 1$, if $j \in P_k$ and $j' \in P_{k'}$, then we turn $P_k$ and $P_{k'}$ into two new bundles $\{j,j' \}$ and $((P_k \cup P_{k'}) \setminus  \{j,j' \})$. Observe that $v_a \left((P_k \cup P_{k'}) \setminus  \{j,j' \}  \right) \geq 1$. Hence, we can assume that for each $B_i \not\in L_a \cup H_a$ (as defined in \eqref{eqn:LKx1}), there exists a $P_{k} = B_i$ and all other bundles value at least 1. Second, we re-enumerate the bundles in $P$ such that 

\begin{itemize}
    \item $P_1, \dots, P_{t_1}$: each has an item $j$ of value more than $\tfrac{5}{8}$ and $v_a(B_j) < \tf34$. 
    \item $P_{t_1+1},\dots, P_{t_2}$: each has an item $j$ of value more than $\tfrac{5}{8}$ and $v_a(B_j) > 1$ (Observe that there are $h_a$ such bundles, so $t_2 - t_1 = h_a$).
    \item $P_{t_3+1},\dots, P_{n}$: each such $P_k = B_{k'}$ for some $k'$ and $\tf34 \leq v_a(B_{k'}) \leq 1$.
    \item $P_{t_2+1}, \dots, P_{t_3}$ be the remaining bundles (Observe that $t_3 = l_a + h_a$).
\end{itemize}

Let $L^1_a\subseteq L_a$ be the set of bags in $L_a$, which have one item with value more than $\tfrac{5}{8}$, and $L^2_a=L_a\setminus L^1_a$. Clearly, $|L^1_a|=t_1$ and $|L^2_a|=l_a-t_1$. 

\begin{observation}\label{obs:last}
 \begin{enumerate}
    \item If $v_a(B_j) > 1 \ge v_{aj} > \tfrac{5}{8}$, then by Corollary \ref{cor:Bound-on-Values}, $v_{a(2n-j+1)} > \tfrac{1}{4}$. 
    \item If $v_a(B_j) < \tf34$ and $v_{aj} > \tfrac{5}{8}$, i.e., $B_j \in L^1_a$ then $v_{a(2n-j+1)} < \tf18$.
    \item Using the two previous observations, if $B_{j_1} \in L^1_a$ and $B_{j_2} \in H_a$, then $j_1 < j_2$ because $v_{a(2n-j_2+1)} > \tfrac{1}{4} > \tf18 > v_{a(2n-j_1+1)}$.
    \item if $B_{j_1} \in H_a$ and $B_{j_2} \in L^2_a$, then $j_1 < j_2$ because $v_{aj_1} > \tfrac{5}{8} \ge v_{aj_2}$.
    \item Bags are \emph{ordered} by value of most-valued item in each bag. First, $L^1_a$ then $H_a$ then $L^2_a$ (excluding the bags for which $\tf34 \le v_a(B_k) \le 1$).
    \item The above observations imply that items with value in $[\tf18,\tfrac{1}{4}]$ (if any) automatically belong to bags $B_j$'s, which are neither in $H_a$ nor $L_a$. Hence, there is no such item in $\Prn$. 
\item After the termination of Algorithm \ref{algo:initial} we have $v_i(S_2) = v_i(B_n) < \tf34 $. Therefore, $B_{n}\in L_a$. Further, since $h_a>0$ (Lemma~\ref{lem:three}), $B_n\in L^2_a$. This implies that $|L^2_a|\ge 1$.
 \end{enumerate}
\end{observation}

Further, observe that there is no item from $M \setminus J$ in $P_{t_3+1},\dots, P_{n}$. We leverage this observation to prove \eqref{eqn:mb1} by only considering partitions $\Prn$ that contain all items of $M \setminus J$. Here is a simple example to help understand the construction.

\begin{example}\label{ex:manipulate}
Consider an example where, after the execution of Algorithm \ref{algo:initial}, $n=6$ (hence $|J|=12$). The valuation of agent $a$ for items in $J$ is shown in Table \ref{tbl:vals}. The construction of bags is shown in Figure \ref{fig:exB} and Table \ref{tbl:bags}.
\begin{table}[h!]
\caption{Valuation of agent $a$ for items in $J$ in Example \ref{ex:manipulate}}\label{tbl:vals}
\centering
\begin{tabular}{||l| c c c c c c c c c c c c||} 
 \hline
  Item & 1 & 2 & 3 & 4&5&6&7&8&9&10&11&12 \\ [0.5ex] 
 \hline\hline
 $v_{aj}$ & $\tfrac{28}{40}$ & $\tfrac{28}{40}$ & $\tfrac{28}{40}$ & $\tfrac{28}{40}$ & $\tfrac{17}{40}$ & $\tfrac{15}{40}$ & $\tfrac{14}{40}$ & $\tfrac{14}{40}$ & $\tfrac{14}{40}$ & $\tfrac{13}{40}$ & $\tfrac{5}{40}$ & $\tfrac{1}{40}$  \\ [1ex] 
 \hline
\end{tabular}
\end{table}

\begin{figure}[h!]
     \centering
     \includegraphics[width=0.8\textwidth]{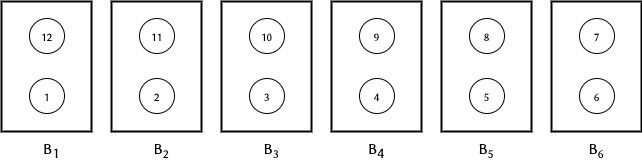}
     \caption{Setting of the items in the bags in Example \ref{ex:manipulate}}
     \label{fig:exB}
 \end{figure}
 
\begin{table}[h!]
\caption{Bags and their valuations in Example \ref{ex:manipulate} where $(H_a \cup L_a)^c$ refers to the bags not in $(H_a \cup L_a)$.}\label{tbl:bags}
\centering
\begin{tabular}{||l| c c c c c c ||} 
 \hline
  Bag & $B_1$ & $B_2$ & $B_3$ & $B_4$& $B_5$ &$B_6$\\ [0.5ex] 
 \hline\hline
 Setting & $\{1,12 \}$ & $\{2,11 \}$ & $\{3,10 \}$ & $\{4,9 \}$ & $\{5,8 \}$ & $\{6,7 \}$   \\ \hline
 $v_{a}(B_k)$ & $\tfrac{29}{40}$ & $\tfrac{33}{40}$ & $\tfrac{41}{40}$ & $\tfrac{42}{40}$ & $\tfrac{31}{40}$ & $\tfrac{29}{40}$  \\  \hline
 Bag type& $L^1_a$ & $(H_a \cup L_a)^c$ & $H_a$ & $H_a$ & $(H_a \cup L_a)^c$ & $L^2_a$\\ [1ex] 
 \hline
\end{tabular}
\end{table}

Now let $P$ be an MMS partition of agent $a$ satisfying the condition of Lemma~\ref{lem:sp} be as Figure \ref{fig:partition}. $M_k \subseteq (M \setminus J)$ is the set of low value items in $P_k$ so $\bigcup_{k \in [6]} M_k = M \setminus J$. $M_k$'s can be possibly empty. Then, the enumerated $P$ is shown in Figure~\ref{fig:pruned}. Observe that: 

\begin{figure}[h!]
     \centering
     \includegraphics[width=0.8\textwidth]{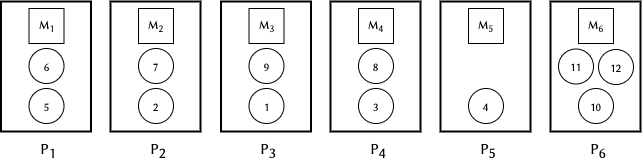}
     \caption{an MMS partition ($P$) of agent $a$ in Example \ref{ex:manipulate}}
     \label{fig:partition}
 \end{figure}

\begin{figure}[h!]
     \centering
     \includegraphics[width=0.8\textwidth]{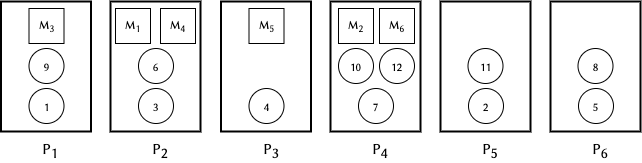}
     \caption{Re-enumarated $P$ in Example \ref{ex:manipulate}}
     \label{fig:pruned}
 \end{figure}
 
\begin{itemize}
    \item $v_a(B_2) \in [\tf34,1]$. Therefore, items 2 and 11 are put into one bundle ($P_5$) and all remaining items, which were previously in the same bundle with each of these items, are in another bundle ($P_4$). This is the case for $B_5$ and items 5 and 8 too. These bundles are put in the last ($P_5$ and $P_6$).
    \item $B_1 \in L^1_a$ and also $v_{a1} > \tfrac{5}{8}$. Therefore, we set the bundle containing item 1 to be $P_1$. $B_1$ is the only bag with this property and hence $t_1 =1$.
    \item $B_3, B_4 \in H_a$. Therefore, we put the bundles containing item 3 and item 4 next ($P_2$ and $P_3$). Therefore, $t_2=3$.
    \item The remaining bundle is set to be next ($P_4$). It means that $t_3=4$.
\end{itemize}
\end{example}

\begin{definition}\label{def:y1}
Define $y:=v_{ak'}$ where $k'$ is the highest value item in the set $\{j\in J_2\cap B_k, B_k\in H_a\}$. 
\end{definition}

Observe that $y$ is the value of the highest valued item from $J_2$ (items in top row of Figure \ref{fig:bags}), which is in a bag with value more than 1. Further, $\tfrac{2}{8} <y < \tfrac{3}{8}$ since by definition $y + v_{ij}>1$ for some $j \in J_1$ and $y = v_{ij'}$ for some $j' \in J_2$. For instance, in Figure \ref{fig:bags}, if all bags in $\{B_{k+1}, \dots, B_{n}\}$ value at most 1 for agent $a$ and $v_a(B_k) >1$, i.e., $B_k$ is in $H_a$, then $y= v_{a(2n-k+1)}$. This also implies that $v_{a1}\ge \dots \ge v_{ak} > \tfrac{5}{8}$. Moreover, $y=v_{a(2n-k+1)} \le v_{a(2n-k)}\le \dots \le v_{a(k+1)}$. In Example~\ref{ex:manipulate}, $y= \tfrac{14}{40}$.

The value $y$ turns out to be crucial, which we use to relate the total value of $\cup_{k=1}^{t_3} P_k \cap (M\setminus J)$ to the value needed in bag filling procedure. Intuitively, if $y$ is high, then less value is needed from $M \setminus J$ to make all bundles in $\Prn$ to be at least 1. This is also true for the bag filling argument: If $y$ is high then less value is needed to add to the bags in \eqref{eq:B_i} to make each of them at least $\tf34$. Otherwise, when $y$ is low, the value of $M \setminus J$ can be shown to be relatively high. We begin with the following lemma.

\begin{lemma}\label{lem:y}
For any bundle $B_k=\{j,j'\}\in L_a$ such that $v_{aj}, v_{aj'} <\tfrac{5}{8}$, we have $v_{aj}, v_{aj'} \ge y$. 
\end{lemma}
\begin{proof}
The claim follows from the construction of $B_k$'s, and the fact that each bundle of $H_a$ has an item with value more than $\tfrac{5}{8}$. Let $k$ be the index of the bundle in \eqref{eq:B_i} for which the $y$ in Definition \ref{def:y1} is achieved. Then {\small $$v_{a1} \ge \dots \ge v_{ak} > \tfrac{5}{8} \ge v_{a(k+1)} \ge \dots \ge v_{a(2n-k+1)} = y \ge v_{a(2n-k+2)} \ge \dots \ge v_{a(2n)}.$$} 
If $v_{aj} < \tfrac{5}{8}$ for $j\in J_1$, then $j > k$ and $j' < 2n-k+1$, implying the result. 
\end{proof}

First, we only consider the case when bags in $L_a$ satisfy Lemma \ref{lem:y}. This is equivalent to the case when $t_1 = 0$. Then, we prove the general case in Theorem \ref{thm:fb}. 

In the bags, which are in $H_a$, there are exactly $2h_a$ items from $J$.  However, in $\{P_{t_1+1},\dots, P_{t_2}\}$, we know that each bundle has one item with value $> \tfrac{5}{8}$ but may contain more items from $J$. We define $z$ to reflect the possible difference, i.e., $z:=\max\{2h_a-|\bigcup_{t=t_1+1}^{t_2} P_t\cap J|, 0\}$. In Lemma \ref{lem:bound-vm-max}, we show two things. First we show a bound for the value that is needed with respect to $y$ to make bags in $L_a$ to become at least $\tf34$ (see definition of $x_i$ in \eqref{eqn:LKx1}). Second, we show a bound for $M \setminus J$ with respect to $z$ and $y$ by using the fact that the value of each bundle in $\Prn$ is at least 1. As we discussed before, we will relate these two observations to prove the bound for $v_a(M \setminus J)$ and we show \eqref{eqn:mb1} is always correct.

\begin{lemma}\label{lem:bound-vm-max}
For agent $a\in N^2$, if $t_1 = 0$ then we have  
\begin{enumerate}
\item $x_a \leq (\tf34-2y)l_a$.
\item $v_a(M\setminus J) \geq \max \{x_a+\tfrac{l_a}{4} - y z+\tfrac{z}{4}, \tfrac{z}{4} \}$, where $z=\max\{2h_a-|\bigcup_{t=t_1+1}^{t_2} P_t\cap J|, 0\}$.
\end{enumerate}
\end{lemma}
\begin{proof}
For the first part, since $t_1=0$, we have $v_a(B_k) \geq 2y, \forall B_k\in L_a$ (Lemma~\ref{lem:y}). This implies that $x_a = (\tf34)l_a - \sum_{k: B_k \in L_a} B_k \le (\tf34 - 2y)l_a$. 

For the second part, in $\bigcup_{k:B_k \in H_a} B_k$, there are exactly $2h_a$ items from $J$. If there are at least $2h_a$ items in $\bigcup_{t=t_1+1}^{t_2} P_t\cap J$ then $z=0$. On the other hand, if there are $2h_a-z$ items in $\bigcup_{t=t_1+1}^{t_2} P_t \cap J$, then it implies that there are at least $z$ bundles in $\{P_{t_1+1},\dots, P_{t_2} \}$ with exactly one item from $J$. Each of these bundles need more than $\tfrac{1}{4}$ value of items from $M\setminus J$ to be 1. 

Next, if $z=0$, it means that there are at least $2h_a$ items from $J$ in $\{P_{t_1+1},\dots, P_{t_2}\}$ and $2l_a$ items from $J$ in $\{P_{t_2+1}, \dots, P_{t_3}\}$ since $t_1 = 0$. Then all bundles in $\{P_{t_2+1}, \dots, P_{t_3}\}$ need at least $x_a+\tfrac{l_a}{4}$ value of items from $M\setminus J$ to become $1$ because the value of each item in the bundles of $L_a$ is at least $y$ (follows from Lemma~\ref{lem:y} and Observation~\ref{obs:last}). If $z>0$, then all bundles in $\{P_{t_2+1}, \dots, P_{t_3}\}$ need at least $x_a+\tfrac{l_a}{4} -yz$ value of items from $M\setminus J$ to become one because each of the $z$ items has value at most $y$ (Lemma~\ref{lem:y}). Therefore, in total, there should be at least $x_a+\tfrac{l_a}{4}- y z+\tfrac{z}{4}$ value of items from $M\setminus J$ in $\{P_{t_1+1}, \dots, P_{t_3}\}$. Letting $Q = P_{t_2+1} \cup \dots \cup P_{t_3}$, we have (using $v(P_t) \ge 1$ and $t_3 - t_2 = l_a$, and abusing notation):\begin{align*}
    l_a &\le v_a(Q)\\
     &= v_a(Q \cap J) + v_a(Q \cap (M\setminus J))\\
     &= v_a(Q \cap H_a) + v_a(Q \cap L_a) + v(Q \cap (M\setminus J))\\
    &\le zy + v_a(L_a) + v_a(Q \cap (M\setminus J))\\
     &= zy + (\tf34)l_a - x_a + v_a(Q \cap (M\setminus J)),
\end{align*}
implying $v_a(Q \cap (M\setminus J)) \ge x_a + \tfrac{l_a}{4} - yz$, where the second-to-last line uses the fact that $|Q \cap J| = 2l_a + z$, so obviously $v_a(Q \cap J) \le v_a(L_a) + zy$ (any items from $H_a$ that end up in $Q$ are worth at most $y$, which is at most the value of any $L_a$ item), and the last line follows from substituting $x_a = (\tf34)l_a - v_a(L_a)$.

Further, in case $z$ is large such that the value of $x_a+\tfrac{l_a}{4}- y z+\tfrac{z}{4}$ is negative, we still have $z$ bundles in $\{P_{t_1+1},\dots, P_{t_2}\}$ with exactly one item (with value of at most $\tf34$) from $J$ in each. Therefore, we need $\tfrac{z}{4}$ value of items from $M\setminus J$ to make all bundles in $\{P_{t_1+1},\dots, P_{t_2}\}$ value 1. 
\end{proof}

In Lemma \ref{lem:bound-vm-max}, we showed that $v_a(M\setminus J) \geq \max \{x_a+\tfrac{l_a}{4} - y z+\tfrac{z}{4}, \tfrac{z}{4} \}$. In Theorem~\ref{thm:maxva1}, we prove \eqref{eqn:mb1} (for the special case when $t_1 = 0$) by relating the two observations from Lemma \ref{lem:bound-vm-max} and using two different cases for $\max \{x_a+\tfrac{l_a}{4} - y z+\tfrac{z}{4}, \tfrac{z}{4} \}$. 
\begin{theorem}\label{thm:maxva1}
For agent $a\in N^2$, if $t_1=0$ then
\begin{equation}
    v_a(M \setminus J) \geq x_a+\tfrac{l_a}{8} \text{ .}
\end{equation}
\end{theorem}

\begin{proof}
Using Lemma~\ref{lem:bound-vm-max}(2) there are two cases for the maximum of $x_a+\tfrac{l_a}{4} - yz+\tfrac{z}{4}$ and $\tfrac{z}{4}$. We prove the claim for each case separately.
\medskip

\noindent\textbf{Case 1:} Suppose $x_a+\tfrac{l_a}{4} - y z+{z}/{4} \geq {z}/{4}$, then 
        $z \leq \frac{x_a+ \tfrac{l_a}{4}}{y}$. Further, Lemma~\ref{lem:bound-vm-max}(1) implies $y\leq \frac{3l_a-4x_a}{8l_a}$. 
   Using these, we get 
    \[v_a(M \setminus J) \geq x_a+ \tfrac{l_a}{4}- z(y-\tfrac{1}{4}) \ge \frac{4x_a+l_a}{16y} \ge \frac{(4x_a+l_a)l_a}{2(3l_a-4x_a)} \enspace .\]
 
Next, consider
        \[\frac{(4x_a+l_a)l_a}{2(3l_a-4x_a)} - \left( x_a+ \frac{l_a}{8}\right)= \frac{28x_a^2+(2x_a-l_a)^2}{8(3l_a-4x_a)}   \geq 0\enspace .\]
The last inequality follows because $(3l_a-4x_a)>0$ (Lemma~\ref{lem:three}). This implies $v_a(M\setminus J) \geq x_a + \tfrac{l_a}{8}$.
\medskip
    
\noindent\textbf{Case 2:} Suppose $x_a+\tfrac{l_a}{4} - y z+\tfrac{z}{4} \le \tfrac{z}{4}$. Then, using Lemma~\ref{lem:bound-vm-max}(1), we have
\[v_a(M \setminus J) \geq \tfrac{z}{4} \geq \frac{4x_a+ l_a}{16y} \geq \frac{(4x_a+l_a)l_a}{2(3l_a-4x_a)} \geq x_a + \tfrac{l_a}{8}\enspace .\qedhere\]
\end{proof}
Next, we handle the general case when $t_1>0$. Let $J_{\min} = \bigcup_{B_k\in L^1_a} B_k \cap J_2$ (recall that $L^1_a$ denote the set of bags in $L_a$ that have one item with value more than $\tfrac{5}{8}$). Then, we have $v_{aj}< \tf18, \forall  j\in J_{\min}$ because these items are bagged with an item value more than $\tfrac{5}{8}$ and together they have value less than $\tf34$. Therefore, $v_a(J_{\min}) < \tfrac{t_1}{8}$. Let us call items in $(M\setminus J)\cup J_{\min}$ as \textit{filler} items (each filler item, by Lemma \ref{lem:three}(iv), has value $<\tf18$ ) and items in $J \setminus J_{\min}$ as \textit{base} items. 

\begin{theorem}\label{thm:fb}
For agent $a\in N^2$, we have $v_a(M \setminus J) \geq x_a+\tfrac{l_a}{8}$. 
\end{theorem}

\begin{proof}
In this proof we only consider the items in $\Prn$ because, as discussed before, these bundles contain all items from $M\setminus J$. Let $x'_a := \frac{3}{4}t_1  - \sum_{B_k\in L^1_a} v_a(B_k)$, and $x''_a := x_a-x'_a$. For agent $a$, we can treat items of $J_{\min}$ as low-value items so we will prove:
 \begin{equation}
     \label{eq-extra-residual1}
     v_a(M \setminus J)+ v_a(J_{\min}) \geq x'_a+\tfrac{t_1}{4}+x''_a+{(l_a-t_1)}/{8} \text{.} 
 \end{equation}
 Since $v_a(J_{\min})$ is at most $\tfrac{t_1}{8}$, \eqref{eq-extra-residual1} directly implies the theorem. 

In $\Prn$, recall that each of $\{P_1, \dots, P_{t_1}\}$ contains a base item whose value is at least the maximum of all the items in the remaining bundles. Further, if there are two base items in any bundle $P_k$ in $\{P_1, \dots, P_{t_1}\}$, then $v_a(P_k\cap J) > 1$ (using Observation~\ref{obs:last}). If each bundle in $\{P_1,\dots, P_{t_1}\}$ has exactly one base item ($t_1$ bundles in total) then items of value at least $x'_a+\tfrac{t_1}{4}$ are needed to make all these bundles at least 1. This, together with Theorem~\ref{thm:maxva1}, proves the bound. On the other hand, if there are more than $t_1$ \textit{base} items in $\{P_1,\dots, P_{t_1}\}$, by Observation \ref{obs:last}(vi), each base item in $\Prn$ has value $> \tfrac{1}{4}$. Therefore, for each extra base item $j$, there are two cases. If $j$ comes from $B_k\in L^2_a$ then clearly, there should be enough items in $M \setminus J$ to compensate for $v_{aj}$ in $\{P_{t_2+1},\dots,P_{t_3}\}$ since the value item $j$ is adding to a bag in $L^1_a$ is more than what it was considered in the right side of \eqref{eq-extra-residual1}. For the other case, if $j$ comes from $B_k\in H_a$ then it will make $B_k\setminus \{j\}$ to need more low-value items than the bundle in $\{P_1,\dots, P_{t_1}\}$ who gets $j$ to become one since if $B_k \in H_a$ and $B_{k'} \in L^1_a$ then $v_{ak} < v_{ak'}$ (Observation \ref{obs:last}(v)). This completes the proof.  
\end{proof}

\section{Algorithm for $\frac{3}{4}$-MMS Allocation}\label{sec:MMS-NOT-known}
The existence proof of $\tf34$-MMS allocation in Section~\ref{sec:MMS-known} requires the knowledge of the exact MMS value $\mu_i$'s of all agents, e.g., the proof of valid reduction for bundle $S_4$ in Lemma~\ref{lem:valid-reduction} needs this assumption. Finding an exact $\mu_i$ of an agent $i$ is an NP-Hard problem, however a PTAS exists~\cite{woeginger1997polynomial}. This implies a PTAS to compute a $(\tf34-\epsilon)$-MMS allocation for an $\epsilon>0$. However, for small $\epsilon$, this PTAS is computationally very expensive and may not be practical. In this section, we show that our algorithmic technique in Section~\ref{sec:MMS-known} is powerful enough to be modified into a strongly polynomial time algorithm to find an exact $\tf34$-MMS allocation for a given instance $\I=\langle N,M,V\rangle$. The key idea is to use the average as an upper-bound for the MMS value; see Lemma~\ref{lem:avg}.

As it is shown in Section \ref{sec:properties}, we assume $\I$ to be an ordered instance, i.e., $v_{i1} \geq v_{i2} \geq \dots \geq v_{im}, \text{ for all } i \in N$. We fix the order of agents in $N$ such that at each step of the algorithm if more than one agent satisfies the conditions of valid reduction, defined in Section~\ref{sec:red}, we choose one with the lowest index. We maintain the same order among remaining agents after every valid reduction. For the ease of exposition, as in Section~\ref{sec:MMS-known}, we abuse notation and use $M$ and $N$ to denote the set of unallocated items and the set of agents who have not received any bundle yet, respectively. Moreover, we use $n:=|N|$ and $m:=|M|$. Further, we use $j$ to denote the $j^{\text{th}}$ highest value item in $M$, and $i$ to denote the $i^{\text{th}}$ agent in $N$ as per the fixed order.

In this section, unlike what we did in Section \ref{sec:MMS-known} (normalizing the valuation of items based on actual $\mu_i^n(M)$), we normalize $v_{ij}$'s using the total value of all items.

\begin{definition}\label{def:norm}
We call an instance $\I=\langle N,M,V\rangle$ \emph{normalized} if $v_i(M)=n$ for all $i \in N$.
\end{definition}

Using the scale invariance property (Lemma \ref{lem:Scale}), we can without loss of generality work with the normalized valuations so that $v_i(M)=n$ for all agents $i \in N$. The proof of following corollary is straightforward. 

\begin{corollary}
For a normalized instance $\I=\langle N,M,V\rangle$, $\mu_i^{n}(M) \le 1 \ \ \forall i.$
\end{corollary}

The algorithm to compute a $\tf34$-MMS allocation is given in Algorithm~\ref{algo-threefourth}. It has three main parts: Initial Assignment, Update Upper Bound, and Bag Filling. Initial Assignment is further divided into two parts: Fixed and Tentative.  
\begin{algorithm}[t!]
\caption{$\tf34$-MMS Allocation}\label{algo-threefourth}
\DontPrintSemicolon
  \SetKwFunction{Define}{Define}
  \SetKwInOut{Input}{Input}\SetKwInOut{Output}{Output}
  \Input{Ordered Instance $\langle N,M,V\rangle$, i.e., $v_{i1} \ge \dots \ge v_{im}, \forall i\in N$}
  \Output{$\tf34$-MMS Allocation}
  \BlankLine
  Normalize Valuations (i.e., scale valuations so that $v_i(M)=n, \forall i\in N$)\tcp*{$n=|N|$}
  $(N,M,V) \gets$ Fixed-Assignment$(N,M,V)$\tcp*{Algorithm~\ref{algo-fixed}}
  $(N^t,M^t,V^t)\gets$ Tentative-Assignment$(N,M,V)$\tcp*{Algorithm~\ref{algo-tentative}}
  \While{$N^{21}\neq \emptyset$} {
    $i \gets$ the lowest index agent in $N^{21}$\tcp*{$N^{21}$ is defined in \eqref{eqn:N21}}
    Undo Tentative-Assignments\tcp*{go back to $(N,M,V)$}\label{line:undo}
  Update $i$'s MMS bound \& normalize valuations\tcp*{Theorem \ref{theo-updatupperbound} } 
  $(N,M,V) \gets$ Fixed-Assignment$(N,M,V)$\tcp*{these assignments are final}
  $(N^t,M^t,V^t)\gets$ Tentative-Assignment$(N,M,V)$\;
  }
  Make all tentative assignments final: $(N,M,V) \gets (N^t,M^t,V^t)$  \;
  Bag-Filling$(N,M,V,\alpha)$ for $\alpha = \tf34$ \tcp*{Section~\ref{sec:bag-filling-UNKNOWN} and Algorithm~\ref{algo:Bag-Filling-MAIN}}
\end{algorithm}

\subsection{Initial Assignment}\label{sec:initial-assignment2}
This section is very similar to Section~\ref{sec:initial-assignment} where we assign $S_1$, $S_2$, $S_3$, and $S_4$ to agents in order to reduce the number of high-value items to at most $2n$. The main difference for this section is assigning $S_4$. Since assigning $S_4$ is a \textit{valid reduction} only if $\mu_i^{n}(M) =1$ (Lemma~\ref{lem:valid-reduction}), therefore, when $S_1$, $S_2$, and $S_3$ value less than $\tf34$ for all agents we check for $S_4$ and if the first $S_4$ is allocated, then this allocation and all initial allocations after that are called \emph{tentative assignments}. Tentative assignment are finalized only if allocation of any $S_4$ did not take away too much valuation for any agent in $N$. 

It is crucial for this step to assign bundles to agents according to an order. Therefore, for $S_1$, $S_2$, $S_3$, and $S_4$ the bundle with lower index has more priority and if there are more than one agent satisfied with $S \in \{S_1,S_2,S_3,S_4 \}$ we choose the one with lower index based on the original order of $N$. This ordering will be useful when we have to undo tentative assignments.

\paragraph{Update $N,M,V$} After every assignment, say bundle $S$ to agent $i$, we update $N, M, V$ as follows: 
        $$M\gets M\setminus S;\ \  N\gets N\setminus\{i\}; \ \ v_{i'j} \gets v_{i'j}\cdot \frac{|N|}{v_{i'}(M)}, \forall i'\in N, j\in M\enspace .$$
    Note that this maintains $\mu_i^{n}(M) \le 1, \forall i\in N$ after normalization (Definition~\ref{def:norm} and Lemma~\ref{lem:avg}).

\subsubsection{Fixed Assignment}\label{sec-fixed}
In this part, we allocate high-value items to agents using Algorithm \ref{algo-fixed}. We assign the bundle $S \in  \{S_1,S_2,S_3\}$ to the lowest index agent $i \in N$ for which $v_i(S)\geq \tf34$. Then, we update $N,M,V$ and repeat this step until no such agent exists. Lemma~\ref{lem:valid-reduction} is applicable here to show that allocating $S\in\{S_1,S_2,S_3\}$ to agent $i$ and removing them from $M$ and $N$ is a \emph{valid reduction}. 
\begin{algorithm}[t!]
\caption{Fixed-Assignment}
\label{algo-fixed}
\DontPrintSemicolon
  \SetKwFunction{Define}{Define}
  \SetKwInOut{Input}{Input}\SetKwInOut{Output}{Output}
  \Input{Ordered Instance $\langle N,M,V\rangle$ (i.e., $v_{i1} \ge \dots \ge v_{im}, \forall i\in N$), where $\mu_i^{n}(M) \le 1, \forall i\in N$}
  \Output{Fixed Assignments and Reduced Instance}
  \BlankLine
  For any $S\subseteq M$, define $\Gamma(S)= \{i \in N: v_i(S) \geq \tf34 \}$\;
  $S_1 := \{1\};\  S_2 := \{n, n+1\};\ S_3 := \{2n-1, 2n, 2n+1\}$\tcp*{bundles that can be assigned}
  $T= \Gamma(S_1)\cup \Gamma(S_2) \cup \Gamma(S_3)$\;
   \BlankLine
     \While{$T \neq \emptyset$}{
     $S\gets$ the lowest index bundle in $\{S_1, S_2, S_3\}$ for which $ \Gamma(S) \neq \emptyset$\;
     $i\gets$ the lowest index agent in $\Gamma(S)$\;
     
     Assign $S$ to agent $i$\tcp*{final assignment}
     Update $N, M, V, T$
     }
  \Return{$N, M, V$}\;
\end{algorithm}

As in Section \ref{sec:MMS-known}, let $J_1 := \{1,\dots,n\}$ denote the set of first $n$ items. Similarly, define $J_2:=\{n+1,\dots,2n\}$ and $J:=J_1 \cup J_2$. Note that Corollary~\ref{cor:Bound-on-Values} is also applicable here. 

\subsubsection{Tentative Assignment}\label{sec-tentative}
This step is a continuation of Fixed Assignment in Section~\ref{sec-fixed}. It starts when Algorithm \ref{algo-fixed} terminates and the first $S_4$ bundle is being assigned. In this part, we assume that $\mu_i^{n}(M)=1,\forall i\in N$, i.e., the actual MMS value for each agent is equal to her current MMS upper bound, and we allocate $S_1,S_2,S_3,S_4$ using this assumption. 
In the next step, either this assumption works fine and we make all tentative assignments final and move to the next stage of bag filling or we detect an agent $i$ for whom $\mu_i^{n}(M)$ is significantly lower than 1. In the latter case, we update $i$'s MMS upper bound. In particular, we update the MMS upper bound by updating the valuations so that the new MMS upper bound remains $1$ for every agent. 

In Algorithm \ref{algo-tentative}, for $S \in \{S_1,S_2,S_3,S_4 \} $ we check whether there exists an $i \in N$ with $v_i(S) \geq \tf34$. If true, then we \emph{tentatively} assign the lowest index such a bundle $S$ to the lowest index such an agent $i$ and we \emph{tentatively} update $M,N,V$ as we did in Section~\ref{sec-fixed}. Choosing the lowest index $S_k$ makes sure that when $S_4$ is assigned, then none of $S_1, S_2, S_3$ satisfies the condition. This is essential in proving that assigning $S_4$ to an agent is a valid reduction in Lemma~\ref{lem:valid-reduction}.  Further, note that when $S_4$ is assigned, the value of each bundle in $\{S_1, S_2, S_3\}$ for every agent is strictly less than $\tf34$. However, after removing an agent $i$ with $S_4$, it may later trigger valid reductions with $\{S_1,S_2,S_3\}$. 

\begin{algorithm}[t!]
\caption{Tentative-Assignment}\label{algo-tentative}
\DontPrintSemicolon
  \SetKwFunction{Define}{Define}
  \SetKwInOut{Input}{Input}\SetKwInOut{Output}{Output}
  \Input{Ordered Instance $\langle N,M,V\rangle$ that satisfies Conditions 1-3 in Corollary \ref{cor:Bound-on-Values}}
  \Output{Tentative Assignments and Reduced Instance}
  \BlankLine
  For any $S\subseteq M$, define $\Gamma(S)= \{i \in N: v_i(S) \geq \tf34 \}$\;
  $S_1 := \{1\};\  S_2 := \{n, n+1\};\ S_3 := \{2n-1, 2n, 2n+1\};\ S_4 := \{1, 2n+1\}$\;
  $T= \Gamma(S_1)\cup \Gamma(S_2) \cup \Gamma(S_3) \cup \Gamma(S_4)$\;
  \BlankLine
     \While{$T \neq \emptyset$}{
     $S\gets$ the lowest index bundle in $\{S_1, S_2, S_3, S_4\}$ for which $ \Gamma(S) \neq \emptyset$\;
     $i\gets$ the lowest index agent in $\Gamma(S)$\;
     Assign $S$ to agent $i$\tcp*{tentative assignment}
     Update $N, M, V, T$\;
     }
  \Return{$N, M, V$}\;
\end{algorithm}

\subsection{Updating MMS Upper Bound} \label{sec-update}
The goal here is first to detect the agents for whom the MMS upper bound is overestimated and second to update this bound for them. To detect these agents we continue as in Section~\ref{sec:MMS-known} by first initializing $n$ bags as in~\eqref{eq:B_i} (see Figure \ref{fig:bags}), and divide agents into two types $N^1$ and $N^2$ according to their valuations for these bags. Recall that $N^1$ is the set of agents whose value for each bag is at most $1$, and $N^2$ is the set of remaining agents. 
 
In this section, we need to analyze agents in $N^2$ more thoroughly. We further partition agents in $N^2$ into two sub-types as follows: 
\begin{equation}\label{eqn:N21}
    N^{21} := \{i\in N^2 : h_i > l_i \text{ and } v_i (M \setminus J) < x_i + l_i / 8\}; \ \ \ \  N^{22} := N^2\setminus N^{21}\enspace . 
\end{equation}

\begin{lemma} 
\label{lem:k<l}
If $h_i \leq l_i$ for $i \in N^2$ we have $v_i (M \setminus J) \geq x_i + \tfrac{l_i}{8}$.
\end{lemma}

\begin{proof}
The value of any bag in $H_i$ is less than $\tf98$ for $i$ (Lemma \ref{lem:three}). By \eqref{eqn:LKx1}, the value of any bag which is not in $H_i \cup L_i$ is at most 1 and the total value of bags in $L_i$ is at most $$ \tf34 l_i - x_i \ .$$ The total value of the items in $M$, in a normalized instance, is $n$ so we have
\begin{equation}\notag
\begin{aligned}
    v_i(M \setminus J) \geq &\ n - \left[(\tf98 h_i) + (n-h_i-l_i) + (\tf34 l_i -x_i) \right] \\
    = & -\tf{h_i}8 + \tf{l_i}4 +x_i \geq x_i+ \tf{l_i}8 \ . 
\end{aligned}\qedhere
\end{equation}
\end{proof}

Lemma \ref{lem:k<l} shows that \eqref{eqn:mb1} holds for all $i \in N^{22}$. Hence, by Lemma \ref{lem:bag-filling-terminates}, if $N^{21}$ is empty, then we do not need to update the MMS upper bounds, and we can proceed to the next stage of bag filling. Note that Theorem~\ref{thm:fb} is still applicable for all agents in $N^2$, which implies the following corollary. 

\begin{corollary}\label{cor:t2B}
For any agent $i\in N^{21}$, $\mu_i<1$.
\end{corollary}

The next theorem gives a new upper bound for these agents.

\begin{theorem}\label{theo-updatupperbound}
For any $i\in N^{21}$, $\mu_i \leq \alpha^* = \max \{\alpha_1,\alpha_2,\alpha_3,\alpha_4,\alpha_5 \}$, where 
\begin{align*}
    \alpha_1 = \tfrac{4}{3} v_{i1} \hspace{4.3cm} & \ \ \ \ \ \alpha_2 = \tfrac{4}{3}(v_{in}+v_{i(n+1)}) & \\ 
    \alpha_3 = \tfrac{4}{3} (v_{i(2n-1)}+v_{i(2n)}+v_{i(2n+1)})  & \ \ \ \ \ \alpha_4 = \tfrac{4}{3} (v_{ik}+v_{ik'}) & \\ 
    \alpha_5 =\max \Big\{\alpha: \tfrac{v_i(M \setminus J)}{\alpha} \ge \sum_{B_k: \tf{v_i(B_k)}{\alpha}<\tf34} & (\tf34 - \tfrac{v_i(B_k)}{\alpha}) +\tf18\big|\big\{B_k: \tfrac{v_i(B_k)}{\alpha}<\tf34\big\}\big| \Big\} \enspace,
\end{align*}
where $k$ and $k'$ are respectively the maximum value items from $J$ and $M\setminus J$ that were not tentatively assigned, and $\langle N,M,V \rangle$ refers to the instance after we undo the tentative assignments in Line~\ref{line:undo} of Algorithm~\ref{algo-threefourth}. 
\end{theorem}

\begin{proof}
For a contradiction, suppose $\mu_i = \beta$ and $\beta > \alpha^*$. We have 
\begin{align*}
    \beta > \alpha_1 &\implies v_{i1} < \tf34\beta \\ 
    \beta > \alpha_2 &\implies (v_{in}+v_{i(n+1)}) < \tf34\beta\\
    \beta > \alpha_3 &\implies (v_{i(2n-1)}+v_{i(2n)}+v_{i(2n+1)})< \tf34\beta \\ 
    \beta > \alpha_4 &\implies (v_{ik}+v_{ik'})< \tf34\beta \enspace .
\end{align*}

From this, we can conclude that $v_i (S) < (\tf34)\mu_i$ for $S \in \{S_1,S_2,S_3 \}$. Therefore, agent $i$ will not be satisfied in Fixed-Assignment step. Further, since each bundle in $\{S_1, S_2, S_3\}$ is worth strictly less than $(\tf34)\mu_i$ to agent $i$, all $S_4$ bundles, which are tentatively assigned, (see Algorithm~\ref{algo-tentative}) have value at most $\mu_i$, and hence all tentative assignments are valid reductions (Lemma~\ref{lem:valid-reduction}). Also, since $\beta > \alpha_4$, no new pair that form $S_4$ will satisfy $i$. Furthermore, by normalizing $v_{ij}$'s (i.e., scaling them by $1/\beta$), the value of the bags in (\ref{eq:B_i}) will increase and consequently $h_i$ will not decrease therefore, $h_i$ will remain greater than $0$. Consequently, since $h_i> 0$ and $\beta >\alpha_5$, agent $i$ will remain in $N^{21}$ even after we normalize her valuation with respect to $\beta$. This is because, by definition of $\alpha_5$, this is the maximum value that we can normalize the valuation such that agent $i$ will not satisfy the definition of $N^{21}$ in \eqref{eqn:N21}. Note that by normalizing the valuations, $x_i$ and $l_i$ will be affected that we have taken into account in definition of $\alpha_5$.

It implies that $\mu_i < \beta$ (Corollary~\ref{cor:t2B}), which is a contradiction.
\end{proof}

\begin{remark}
It is easy to see from \eqref{eqn:N21} that if we scale the valuations by $\tfrac{1}{\alpha}$ for $\alpha = v(M \setminus J)/(x_i + \tfrac{l_i}{8})$, then agent $i$ will not remain in $N^{21}$. However, $\alpha_5 $ can be more than $\alpha$ because if we scale the valuations by $\tfrac{1}{\alpha}$ for $0 < \alpha <1$ both $x_i$ and $l_i$ will decrease (see \eqref{eqn:LKx1}). 
\end{remark}

\begin{remark}
$\alpha_5$ can be computed in $O(n)$ time using a simple procedure. Sort the value of bags in $L_i$. Let $u^k$ be the value of the $k^{\text{th}}$ highest value bag in $L_i$. If $\alpha \in (1, u^1/(\tf34))$ and we scale the valuations by $\tfrac{1}{\alpha}$ then all bags, which were in $L_i$ will remain in $L_i$. Similarly, if $\alpha \in [u^k/(\tf34), u^{k+1}/(\tf34))$ then $l_i$ will be reduced by $k$. After figuring out which bags remain in $L_i$ and which bags do not, we can compute $x_i$ with respect to $\alpha$. Then, we can check whether there exists an $\alpha$ in the range for each $k$, starting with $k=1$, that holds the definition of $\alpha_5$. If not, then we move to the next $k$.
\end{remark}

\begin{example}\label{ex:N21}
Consider an example where, after the execution of Algorithm \ref{algo-tentative}, $n=30$ (hence $|J|=60$). The valuation of some agent $i \in N^{21}$ for items in $J$ is shown in Table \ref{tbl:vals2} and her valuation for bags is shown in Table \ref{tbl:bags2}. In this example, we show how to find $\alpha_5$.
\begin{table}[h!]
\caption{Valuation of agent $i$ for items in $J$ in Example \ref{ex:N21}}\label{tbl:vals2}
\centering
\begin{tabu}to0.8\linewidth{||l| X[c] X[c] X[c] X[c] X[c]  ||}
 \hline
  Item & $1-27$ & $28-30$ & $31$ & $32$&$33-60$ \\ [0.5ex] 
 \hline\hline
 $v_{ij}$ & $\tf{29}{40}$ & $\tf{15}{40}$ & $\tf{14}{40}$ & $\tf{13}{40}$ & $\tf{12}{40}$   \\ [1ex] 
 \hline
\end{tabu}
\end{table}
 
\begin{table}[h!]
\caption{The value of the bags in Example \ref{ex:N21}.}\label{tbl:bags2}
\centering
\begin{tabu}to0.8\linewidth{||l| X[c] X[c] X[c] X[c] ||}
 \hline
  Bag & $B_1-B_{27}$ & $B_{28}$ & $B_{29}$ & $B_{30}$\\ [0.5ex] 
 \hline\hline
 $v_{i}(B_k)$ & $\tfrac{41}{40}$ & $\tfrac{27}{40}$ & $\tfrac{28}{40}$ & $\tfrac{29}{40}$   \\  \hline
 Bag type& $H_i$ & $L_i$ & $L_i$ & $L_i$ \\ [1ex] 
 \hline
\end{tabu}
\end{table}

Let $v_i(M \setminus J) = \tf9{40}$. First, observe that $h_i = 27, l_i=3,x_i = \tf6{40},$ and $v_i(M\setminus J)= \tf9{40} < \tf{21}{40} = \tf{l_i}8 + x_i  $. It shows that agent $i$ belongs to $N^{21}$. 

In order to obtain $\alpha_5$, we check how the value of $l_i$ would change after scaling the valuations of agent $i$ by $\tf1{\alpha_5}$. \begin{itemize}
    \item Let $\alpha_5 > \tf{29}{40}$ then $l_i = 3$. We check if there exists an $\alpha_5$ in this range for which the following bound holds.  \begin{equation}
        \label{eq:N21a}
        \tf1{\alpha_5}\times v_i(M\setminus J) \ge \tf38+ (\tf34 - \tf{v_i(B_{28})}{\alpha_5})+(\tf34 - \tf{v_i(B_{29})}{\alpha_5})+(\tf34 - \tf{v_i(B_{28})}{\alpha_5})
    \end{equation} 
    \item If not, let $ \tf{28}{40} < \alpha_5 \le \tf{29}{40}$ then $l_i = 2$ ($B_{30}$ will not belong to $L_i$ anymore). We check if there exists an $\alpha_5$ in this range for which the following bound holds. \begin{equation}
        \label{eq:N21b}
        \tf1{\alpha_5}\times v_i(M\setminus J) \ge \tf28+ (\tf34 - \tf{v_i(B_{28})}{\alpha_5})+(\tf34 - \tf{v_i(B_{29})}{\alpha_5})
    \end{equation}  
    \item If not, $ \tf{27}{40} < \alpha_5 \le \tf{28}{40}$ then $l_i = 1$. We check if there exists an $\alpha_5$ in this range for which the following bound holds \begin{equation}
        \label{eq:N21c}
        \tf1{\alpha_5}\times v_i(M\setminus J) \ge \tf18+ (\tf34 - \tf{v_i(B_{28})}{\alpha_5})
    \end{equation}  
    \item Else, $\alpha_5 = \tf{27}{40}$ ($l_i =0$).
\end{itemize}

Note that in \eqref{eq:N21a}, \eqref{eq:N21b}, and \eqref{eq:N21c} we have obtained the sum of $\tf{l_i}8$ and the new $x_i$ after scaling the valuations.
It takes at most $l_i$ steps (3 steps in this example) to find the maximum $\alpha_5$ for which one of the bounds above holds.
\end{example}

We pick an agent $i$ with the lowest index in $N^{21}$ and update $i$'s valuation as $v_{ij} \gets \frac{v_{ij}}{\alpha}, \forall j \in M$ and repeat. In Theorem~\ref{thm:rt}, we show that the number of updates is at most $n^3$ in the entire run of algorithm. 

\subsection{Bag Filling and Running Time}\label{sec:bag-filling-UNKNOWN}
Note that we reach the bag filling stage only when $N^{21}$ (see~\eqref{eqn:N21}) is empty. Here, we use Algorithm~\ref{algo:Bag-Filling-MAIN} of Section~\ref{sec:Bag-Filling-Known}. Since $N^{21}$ is empty, Lemma~\ref{lem:bag-filling-terminates} is applicable, which implies a $\tf34$-MMS allocation. Next, we show that the entire algorithm runs in strongly polynomial time. 

\begin{theorem}\label{thm:rt}
The entire algorithm runs in $O(n^5m)$ time for an ordered instance, and in $O(nm(n^4 + \log{m}))$ time for any instance.
\end{theorem}

\begin{proof}
In Algorithm~\ref{algo-threefourth}, normalization takes $O(mn)$ arithmetic operations. Both Fixed-Assignment and Tentative-Assignment procedures take at most $O(n^2m)$ arithmetic operations. The bag filling procedure takes at most $O(n^2m)$ time. Each iteration of the while loop takes $O(n^2)$ arithmetic operations to find the set $N^{21}$ of agents, $O(m)$ arithmetic operations to update valuations of an agent $i\in N^{21}$, and $O(n^2m)$ arithmetic operations to run Fixed-Assignment and Tentative-Assignment procedures. Therefore, each iteration of the while loop takes at most $O(n^2m)$ arithmetic operations.

To bound the number of iterations of the while loop, we upper bound the number of times it is run for a particular agent, say $a$, in $N^{21}$. Consider the first iteration when $a$ is the lowest index agent, then if we update $a$'s valuation due to $\alpha \in \{\alpha_1, \alpha_2, \alpha_3\}$ (see Section~\ref{sec-update}), then $a$ gets a fixed assignment after this iteration, and we will not see her again. If we update $a$'s valuation due to $\alpha = \alpha_4$, then in the future iterations, agent $a$ will not be in $N^{21}$ unless the instance is reduced due to a fixed assignment. This could happen only when another agent in $N^{21}$ ends her iteration due to $\alpha \in \{\alpha_1, \alpha_2, \alpha_3\}$. Clearly, this can occur at most $O(n)$ time. Finally, if we update $a$'s valuation due to $\alpha=\alpha_5$, then agent $a$ will not be in $N^{21}$ again unless the instance is reduced due to either fixed or tentative assignment, which can affect the set of agents in $N^{21}$. This could happen only when another agent of $N^{21}$ ends her iteration due to $\alpha\in\{\alpha_1,\alpha_2,\alpha_3,\alpha_4\}$. Clearly, this can occur at most $O(n^2)$ time. Therefore, the maximum number of iterations of the while loop is at most $O(n^3)$. Since each iteration takes at most $O(n^2m)$ arithmetic operations, Algorithm~\ref{algo-threefourth} takes $O(n^5m)$ time. Note that this is for an ordered instance.

The reduction from the general instance to ordered instance is given in~\ref{sec:mp}. Algorithm~\ref{algo:Order} takes $O(mn\log{m})$ time to make the original instance ordered by creating a sorted list of items for each agent. Next, we show that Algorithm~\ref{algo:Reorder} can be implemented in $O(mn)$ time. Initialize a binary array $A$ of size $m$ to all ones. $A[j]$ indicates whether item $j$ is available or not; 1 means available. Each agent has a sorted list of items from the Algorithm~\ref{algo:Order}. From the solution of the ordered instance using Algorithm~\ref{algo-threefourth}, we can create an array $B$ of size $m$, where $B[j]$ stores the agent who is assigned the item $j$. Observe that $B$ can be constructed in $O(m)$ time. Then, for each item $j$ from $1$ to $m$, we get the agent $B[j]$ who is assigned $j$, and then we try to give $B[j]$ the highest item in her list. If her highest item is not available, which we can check from $A$, then we delete this item from $B[j]$'s list, and move down to the next highest item and repeat until we reach at an available item. We give this item to $B[j]$ and mark it assigned in $A$. Since the size of each agent's list is $m$, the total time is $O(mn)$. Using Algorithms~\ref{algo:Order} and ~\ref{algo:Reorder} together with Algorithm~\ref{algo-threefourth}, the running time of the entire procedure is $O(nm(n^4+\log{m}))$.  
\end{proof}

\section{Existence of $(\frac{3}{4}+\frac{1}{12n})$-MMS Allocation}\label{sec:last}
In this section, we show that our approach in Section~\ref{sec:MMS-known} can be extended to obtain the existence of a ($\frac{3}{4}+\gamma$)-MMS allocation for any given instance $\I=\langle N,M,V\rangle$ where $\gamma=\frac{1}{12n}$. We note that $\gamma$ is a constant for the given instance, where $n:=|N|$. We assume that the MMS value $\mu_i$ of each agent $i$ is given. Finding an exact $\mu_i$ is an NP-Hard problem, however a PTAS exists~\cite{woeginger1997polynomial}. This implies a PTAS to compute a $(\frac{3}{4}+\gamma-\epsilon)$-MMS allocation for any $\epsilon>0$. Using the properties shown in Section~\ref{sec:properties}, we normalize valuations so that $\mu_1=1, \forall i$ and assume that $\I$ is an ordered instance, i.e., $v_{i1} \geq \dots \geq v_{i|M|}, \forall i$. Our proof is algorithmic. If more than one agent satisfies the conditions \eqref{eqn:red} of valid reduction, then we choose one arbitrarily.

For the ease of exposition, we abuse notation and use $M$ and $N$ to denote the set of unallocated items and the set of agents who have not received any bundle yet, respectively.  Further, we use $j$ to denote the $j^{\text{th}}$ highest value item in $M$. Moreover, we use $n:=|N|$ and $m:=|M|$. This is the reason we use $\gamma$, which is a constant for a given instance, to denote the approximation factor.

The approach is identical to Section~\ref{sec:MMS-known}. Here, we run Algorithm~\ref{algo:first} with $\alpha=\FRAC$. The analysis is also almost same. Here, we crucially use the extra $\tf18$ (see \eqref{eqn:mb1} and the paragraph after it) to obtain a better factor. To avoid repetition, we only highlight the differences. 

\subsection{Initial Assignment}\label{sec:f-initial-assignment}
In Algorithm \ref{algo:initial}, we keep assigning a bundle $S\in \{S_1,S_2,S_3,S_4\}$ to agent $i$, if any, for which $v_i(S)\geq \FRAC$. Then, we update $M$ and $N$ to respectively reflect the current unallocated items and agents who are not assigned any bundle yet. Observe that Lemma~\ref{lem:valid-reduction} is applicable here to show that assigning a bundle $S\in\{S_1,S_2,S_3\}$ is a \emph{valid reduction}. However, it only applies for $S_4$ when $v_i(S_4) \le 1,  \forall i$. 

As in Section \ref{sec:MMS-known}, let $J_1 := \{1,\dots,n\}$ denote the set of first $n$ items. Similarly, define $J_2:=\{n+1,\dots,2n\}$ and $J:=J_1 \cup J_2$. The following corollary is straightforward.

\begin{corollary}\label{cor:f-Bound-on-Values}
If $v_i(S)< \FRAC, \forall i$ and $\forall S \in\{S_1,S_2,S_3\}$, then $(i)$ $v_{ij} < \FRAC,\ \forall j \in J_1$,\ $(ii)$ $v_{ij} < \frac{3}{8}+\frac{\gamma}{2},\ \forall j \in J_2$, and $v_{in} <\FRAC - v_{i(n+1)}$, and $(iii)$ $v_{ij} < \frac{1}{4} + \frac{\gamma}{3},\ \forall j \in M\setminus J$, for all $i$.
\end{corollary}

However, the value of $S_4$ might be strictly greater than 1 for an agent $i$. Corollary~\ref{cor:f-Bound-on-Values} implies that, $v_i(S_4) < \frac{3}{4}+ \gamma +\frac{1}{4}+\frac{\gamma}{3} = 1+\frac{4\gamma}{3} .$
We use dummy items to fix the issue of extra $\frac{4\gamma}{3}$ lost in assigning $S_4$. 

\paragraph{Dummy items ($D_1$):} For each removed $S_4$ in Algorithm \ref{algo:initial}, we add one dummy item $d_j$ such that 
\begin{equation}\label{eq:dummy-value}
    v_i(\{d_j\}) = \frac{4\gamma}{3} \ \ \ \forall i \in N\enspace .
\end{equation}  
Let $D_1$ denote the set of all dummy items after the termination of Algorithm \ref{algo:initial}. We note that dummy items will not be assigned to any agent. They are defined to make proofs easier. Later, we introduce two more sets $D_2$ and $D_3$ of dummy items in Section~\ref{sec:f-Bag-Filling-Known}. Let $D:= D_1 \cup D_2 \cup D_3$ with $D_2=D_3=\emptyset$ currently. The proof of the following corollary easily follows from Lemma~\ref{lem:valid-reduction}.

\begin{corollary}\label{cor:MMS+Dummy}
For any remaining agent $i\in N$, $\mu_i^{n}(M\cup D) \ge 1$.
\end{corollary}

The proof of the following lemma easily follows using \eqref{eq:dummy-value} and Corollary~\ref{cor:MMS+Dummy}.

\begin{lemma}\label{lem:1/9n-MMS-Decrease}
When Algorithm \ref{algo:initial} terminates, we have $\mu_i^n(M) \ge 1 - \frac{4\gamma r}{3}$ and $v_i(M) \ge n - \frac{4\gamma r}{3}, \forall i$, where $r$ is the total number of rounds in Algorithm \ref{algo:initial} when $S_4$ is assigned.
\end{lemma}

\subsection{Bag Filling}\label{sec:f-Bag-Filling-Known}
The overall approach is same as in Section~\ref{sec:Bag-Filling-Known}. We use the bag filling procedure given in Algorithm~\ref{algo:Bag-Filling-MAIN} to satisfy the remaining agent by setting $\alpha=\FRAC$. We initialize the bags with items in $J$ as in \eqref{eq:B_i} (see Figure~\ref{fig:bags}). Algorithm~\ref{algo:Bag-Filling-MAIN} has $n$ rounds. In each round $k$, it starts a new bundle $T$ with $T \gets B_k$. If there is an agent who values $T$ to be at least $\FRAC$, then assign $T$ to such an agent. Otherwise, keep adding items from $M\setminus J$ to $T$ one by one until someone values $T$ at least $\FRAC$.

For correctness, we need to show that there are enough items in $M\setminus J$ to add on top of each bag in~\eqref{eq:B_i} so that each agent gets a bundle that they value at least $\FRAC$. For this, we first divide agents into two types: $N^1:=\{i\in N\ |\ v_i(B_k) \le 1+\frac{3\gamma}{2}, \forall k\}$ and $N^2 := N\setminus N^1$. Next, we update the notations of \eqref{eqn:LKx1} to reflect the improved bound. For an agent $i\in N^2$, define
\begin{equation}\label{eqn:LKx}
\begin{aligned}
    L_i := & \ \{B_k: v_i(B_k) < \FRAC\}; \ \ \ \ l_i:=|L_i|\\  
    H_i := & \{B_k: v_i(B_k) > 1 + \tfrac{3\gamma}{2}\}; \ \ \ \ h_i := |H_i|\\
    x_i := & \ (\FRAC)l_i - \sum_{k:B_k \in L_i} v_i(B_k). 
\end{aligned}
\end{equation}

The proof of the following lemma is an easy extension of Lemma~\ref{lem:three}, hence omitted.

\begin{lemma}\label{lem:f-three}
For an agent $i\in N^2$, $(i)$ $l_i >0$ and $h_i > 0$, $(ii)$ $v_{i1} > \tfrac{5}{8}+\gamma$, $(iii)$ $v_i(B_k) <\tfrac{9}{8} + \tfrac{3\gamma}{2}, \forall k$, and $(iv)$ $v_{ij} < \tf18, \forall j \in M\setminus J$.
\end{lemma}

The proof of the following lemma is an extension of the proof of Lemma~\ref{lem:bag-filling-terminates}, and is in~\ref{sec:mp}.

\begin{restatable}{lemma}{termination}\label{lem:f-bag-filling-terminates}
If $v_i(M\setminus J) \ge x_i + \tfrac{l_i}{8} -\tf18, \forall i\in N^2$, then Algorithm~\ref{algo:Bag-Filling-MAIN} with $\alpha=\FRAC$ gives every agent a bundle that they value at least $\FRAC$.
\end{restatable}

Now, we only need to show that for each $i\in N^2$, we have
\begin{equation}\label{eqn:mb}
     v_i(M\setminus J) \ge x_i + \tfrac{l_i}{8} -\tf18
\end{equation} 

We start with a few lemmas. Recall from \eqref{eq:MMSp} that $\MP^n_{i}(M)$ denote the set of partitions achieving $\mu_{i}$. The proofs of the following two lemmas are extensions of the corresponding Lemmas~\ref{lem:Giver-Bundle} and \ref{lem:sp} respectively, and are given in~\ref{sec:mp}.

\begin{restatable}{lemma}{giver}\label{lem:f-Giver-Bundle}
For an agent $i\in N^2$, there exists a bundle $P_k$ in every partition $P =\{P_1, \dots, P_n\} \in \MP^{n}_{i}(M\cup D)$ such that $v_i(P_k \setminus J) > \tfrac{1}{4} - \gamma$.
\end{restatable}

\begin{restatable}{lemma}{fsp}\label{lem:f-sp}
For every agent $i$, there exists a partition $P = \{P_1, \dots, P_n\} \in \MP^{n}_{i}(M\cup D)$ such that each $P_k\in P$ has at most one item $j$ with $v_{ij} > \tfrac{5}{8}+\gamma$. 
\end{restatable}

Fix an agent, say $a$, in $N^2$. For simplicity, until the end of this section, when we use \textit{value} of an item or a bundle, we mean their value for agent $a$ (unless mentioned otherwise). Recall that we need to show \eqref{eqn:mb}. Let $P=\{P_1, \dots, P_n)\in \MP^n_a(M\cup D)$ be a partition satisfying the condition of Lemma~\ref{lem:f-sp}. 

We manipulate $P$ as follows: First, for each $B_i= \{j,j' \}$ with $\FRAC \leq v_a(B_i) \leq 1+\tfrac{3\gamma}{2}$, if $j \in P_k$ and $j' \in P_{k'}$, then we turn $P_k$ and $P_{k'}$ into two new bundles $\{j,j' \}$ and $(P_k \cup P_{k'} \cup \{d_j\} \setminus  \{j,j' \})$, where $d_j$ is a new dummy item with $v_a(d_j):=\tfrac{3\gamma}{2}$. 

\paragraph{Dummy items ($D_2$):} Let $D_2$ denote the set of all dummy items added during the manipulation of $P$. Observe that $|D_2| = n - h_a - l_a$. Each $d_j \in D_2$ values $\tfrac{3\gamma}{2}$ for agent $a$. Therefore,
\begin{equation}\label{eq:D2_total}
     v_a(D_2) = \frac{3\gamma|D_2|}{2} \enspace . 
\end{equation}

Later we introduce one more set $D_3$ of dummy items. Recall that $D= D_1 \cup D_2 \cup D_3$ where $D_1$ is defined in \eqref{eq:dummy-value} and $D_3 = \emptyset$ currently. Observe that $v_a \left(P_k \cup P_{k'} \cup \{d_j\} \setminus  \{j,j' \}  \right) \geq 1$. Hence, we can assume that for each $B_i \not\in L_a \cup H_a$ (as defined in \eqref{eqn:LKx}), there exists a $P_{k} = B_i$ and all other bundles value at least 1. Second, we re-enumerate the bundles in $P$ such that 

\begin{itemize}
    \item $P_1, \dots, P_{t_1}$: each has an item $j$ of value more than $\tfrac{5}{8}+\gamma$ and $v_a(B_j) < \FRAC$. 
    \item $P_{t_1+1},\dots, P_{t_2}$: each has an item $j$ of value more than $\tfrac{5}{8}+\gamma$ and $v_a(B_j) > 1$ (Observe that there are $h_a$ such bundles, so $t_2 - t_1 = h_a$).
    \item $P_{t_3+1},\dots, P_{n}$: each such $P_k = B_{k'}$ for some $k'$ and $\FRAC \leq v_a(B_{k'}) \leq 1+\tfrac{3\gamma}{2}$.
    \item $P_{t_2+1}, \dots, P_{t_3}$ be the remaining bundles (Observe that $t_3 = l_a + h_a$).
\end{itemize}

\begin{definition}\label{def:y}
Define $y:=v_{ak'}$ where $k'$ is the highest value item in the set $\{j\in J_2\cap B_k, B_k\in H_a\}$. 
\end{definition}
Observe that $\tfrac{1}{4} + \tfrac{\gamma}{2} < y < \tfrac{3}{8} + \tfrac{\gamma}{2}$. For instance, in Figure \ref{fig:bags}, if no bag from $\{B_{k+1}, \dots, B_{n}\}$ is in $H_a$ but $B_k \in H_a$, then $y= v_{a(2n-k+1)}$. This also implies that $v_{a1}\ge \dots \ge v_{ak} > \tfrac{5}{8} + \gamma$. Moreover, $y=v_{a(2n-k+1)} \le v_{a(2n-k)}\le \dots \le v_{a(k+1)}$. The proof of the following lemma is an easy extension of the proof of Lemma~\ref{lem:y}, and hence omitted. 

\begin{lemma}\label{lem:f-y}
For any bundle $B_k=\{j,j'\}\in L_a$ such that $v_{aj}, v_{aj'} <\tfrac{5}{8}+\gamma$, we have $v_{aj}, v_{aj'} \ge y$. 
\end{lemma}

Next, we show the extension of Lemma~\ref{lem:bound-vm-max}, whose proof is given in~\ref{sec:mp}.

\begin{restatable}{lemma}{fvmax}\label{lem:f-bound-vm-max}
For agent $a\in N^2$, if $t_1 = 0$ then we have  
\begin{enumerate}
\item $x_a \leq (\FRAC-2y)l_a$
\item $v_a((M\cup D)\setminus J) \geq \max \{x_a+l_a(\tfrac{1}{4}-\gamma) - y z+{z}(\tfrac{1}{4}-\gamma),  {z}(\tfrac{1}{4}-\gamma)\}$, where $z=\max\{2h_a-|\bigcup_{t=t_1+1}^{t_2} P_t\cap J|, 0\}$.
\end{enumerate}
\end{restatable}

The following lemma is an extension of Theorem~\ref{thm:maxva1}, and its proof is given in~\ref{sec:mp}.

\begin{restatable}{theorem}{maxva}\label{thm:maxva}
For agent $a\in N^2$, if $t_1=0$ then
\begin{equation}
    v_a(M \cup D\setminus J) \geq x_a+\tfrac{l_a}{8} \text{ .}
\end{equation}
\end{restatable}

Next, we handle the general case when $t_1>0$. Let $L^1_a\subseteq L_a$ be the set of bags in $L_a$, which have one item with value more than $\tfrac{5}{8}+\gamma$, and $L^2_a=L_a\setminus L^1_a$. Clearly, $|L^1_a|=t_1$ and $|L^2_a|=l_a-t_1$. Note that $|L^2_a|\ge 1$ because $B_{n}\in L^2_a$.

\paragraph{Dummy items ($D_3$):} For each bag in $L^1_a$ we add a dummy item $d_j$ where $v_a(d_j) = \gamma$. Therefore,
\begin{equation}\label{eq:D3_total}
     v_a(D_3) = |D_3|\cdot \gamma = |L^1_a|\cdot \gamma < l_a\cdot \gamma \enspace .
\end{equation}
The following theorem is an extension of Theorem~\ref{thm:fb}, and its proof is given in~\ref{sec:mp}.

\begin{restatable}{theorem}{thmfb}\label{f-thm:fb}
For agent $a\in N^2$, we have $v_a((M\cup D) \setminus J) \geq x_a+\frac{l_a}{8}$. 
\end{restatable}

Next, we obtain the bound without the dummy items in the next theorem.

\begin{theorem}\label{theo:type2-enough}
For any agent $a\in N^2$, $v_a(M\setminus J) \ge x_a + \tf{l_a}{8} -\tf18$.
\end{theorem}

\begin{proof}
From Lemma~\ref{lem:1/9n-MMS-Decrease}, \eqref{eq:D2_total}, and \eqref{eq:D3_total}, the total value of the items in $D$ is at most $$v_a(D) \le \frac{4\gamma r}{3} + \frac{3\gamma|D_2|}{2} + \gamma l_a \le \frac{3\gamma n}{2} \le  \frac{1}{8},$$ because $|D_2|+r \le (n-h_a-l_a)$ and $\gamma = \frac{1}{12n}$, where $n$ is the number of agents in the original instance. This, together with Theorem \ref{f-thm:fb}, proves the theorem. 
\end{proof}

\begin{lemma}
The bound $\tf34+\tfrac{1}{12n}$ is tight. 
\end{lemma}
\begin{proof}
The extra $\tfrac{1}{12n}$ is obtained from the extra $\tf18$ in $v_i(M \setminus J)$ (see Lemma~\ref{lem:bag-filling-terminates} and Theorem \ref{thm:fb}), so each agent can take at most $\tfrac{1}{8n}$ extra. For any bag $B_k$ we have $v_i(B_k) < \tf34+\gamma + \tfrac{1}{4}+ \tfrac{\gamma}{2} = \tfrac{9}{8} + \tfrac{3\gamma}{2} = \tfrac{1}{8n}$, which has exactly $\tfrac{1}{8n}$ more than the maximum amount that a bag can have in analysis of Section \ref{sec:MMS-known}. $|H_i|$ can be as big as $n-1$, and hence the bound is tight.
\end{proof}

\section{Conclusions}
We developed a new approach that gives a simple algorithm for showing the existence of a $\tf34$-MMS allocation. Furthermore, we showed that our approach is powerful enough to be easily extended to obtain $(i)$ a strongly polynomial time algorithm to find a $\tf34$-MMS allocation, and $(ii)$ the existence of a $(\tf34+\tfrac{1}{12n})$-MMS allocation, improving the best previous factor. Consequently, this gives a PTAS for finding a ($\tf34+\tfrac{1}{12n}-\epsilon$)-MMS allocation for any $\epsilon>0$. The bound $\tf34+\tfrac{1}{12n}$ is in fact tight.

It could be worth exploring whether the approach extends to obtaining a $\tfrac{4}{5}$-MMS allocation. Such an extension would be challenging because after the initial greedy assignments, there will be $3n$ high-value items, and this would make the process of initializing the bag filling procedure harder due to too many items to handle and also the value of some bags might exceed significantly more than 1.
\bigskip

\noindent{\bf Acknowledgments.}
We would like to thank anonymous referees for their comments and suggestions that have helped to improve the presentation of the paper. Work on this paper supported by NSF Grant CCF-1942321 (CAREER). 

\bibliographystyle{abbrv}
\bibliography{bibliography}

\appendix
\section{Missing Proofs}\label{sec:mp}
\Average*
\begin{proof}
Suppose for some agent $i$, $\mu_i^n(M) > v_i(M)/n$, i.e., there exists a partition of $M$ into $n$ bundles where all bundles have value strictly more than $\frac{v_i(M)}{n}$. Therefore, $v_i(M) \geq n \cdot \mu_i^n(M) > n\cdot v_i(M)/n = v_i(M)$, which is a contradiction.
\end{proof}

\Scale*
\begin{proof}
For any bundle $S\subseteq M$, we have $v_i'(S) = c_i\cdot v_i(S)$. Therefore, $\mu_i' = c_i\cdot \mu_i$. 
Further, $v_i'(A_k) = c_i\cdot v_i(A_k) \ge c_i\cdot\alpha\cdot \mu_i = \alpha\cdot \mu_i', \forall k$. 
\end{proof}

\Order*
\begin{proof}
Consider Algorithms~\ref{algo:Order} and~\ref{algo:Reorder}. It is enough to show that given any instance $\I=\langle N,M,V \rangle$, we can find an ordered instance $\I'=\langle N,M,V' \rangle$ in polynomial time using Algorithm~\ref{algo:Order}. Furthermore, given an $\alpha$-MMS allocation $A'$ for the ordered version $\I'$, we can find an $\alpha$-MMS allocation $A$ for the original instance $\I$ in polynomial time using Algorithm~\ref{algo:Reorder}.

\begin{algorithm}[tbh!]
\caption{\cite{barman2017approximation} Conversion to an Ordered Instance}
\label{algo:Order}
  \SetKwFunction{Define}{Define}
  \DontPrintSemicolon
  \SetKwInOut{Input}{Input}\SetKwInOut{Output}{Output}
  \Input{Instance $\I=\langle N,M,V\rangle$}
  \Output{Instance $\I'=\langle N,M,V'\rangle$ in which $v'_{i1} \geq \dots \geq v'_{im}$ for all $i \in N$}
  \BlankLine
  \For{$i= 1$ to $n$}{
  \For{$j= 1$ to $m$}{
  $j^* \gets j^{\text{th}}$ highest value item of $M$ for agent $i$\;
  $v'_{ij} \gets v_{ij^*}$\;}
  }
\end{algorithm}
\begin{algorithm}[tbh!]
\caption{\cite{barman2017approximation} $\alpha$-MMS Allocation for Unordered Instance}
\label{algo:Reorder}
  \SetKwFunction{Define}{Define}
  \DontPrintSemicolon
  \SetKwInOut{Input}{Input}\SetKwInOut{Output}{Output}
  \Input{Instance $\I=\langle N,M,V\rangle$, Ordered Instance $\I'=\langle N,M,V'\rangle$, $\alpha$-MMS allocation $A'=(A'_i)_{i\in N}$ for $\I'$  }
  \Output{$\alpha$-MMS allocation $A=(A_i)_{i\in N}$ for $\I$}
  \BlankLine
  $A_i \gets \emptyset, \forall i$\;
  \For{$j= 1$ to $m$}{
  $i^* \gets i \in N: j \in A'_{i}$ \tcp*{$i^*$ is the agent who has $j^{\text{th}}$ item of $\I'$ in her bundle}
  $j^* \gets \arg\max_{j \in M \setminus (\bigcup_{i \in N} A_i) } v_{i^*j}$ \tcp*{$j^*$ is $i^*$'s favorite unassigned item}
  $A_{i^*} \gets A_{i^*} \cup \{j^* \} $}
\end{algorithm}

Clearly, items in $\I'$ are sorted by their values, and they have the same order for all agents. Also, it takes $mn$ iterations to obtain $\I'$. Note that each agent's MMS value will remain the same in $\I'$ because it neither depends on the order of items nor on the valuations of other agents. We prove that $v_i(A_i) \geq v'_i(A_i')$ for all $i \in N$. Consider the round $r$ of the Algorithm \ref{algo:Reorder} and consider item $r$ from $\I'$ is in $A'_i$. It means that agent $i$ is getting her $r^{\text{th}}$ favorite item in $A'_i$ but now after $r-1$ round, exactly $r-1$ items are allocated and she will get her favorite item among unallocated items. Therefore, the item she gets in this round is at least as valuable as the one she was getting in the ordered instance. Therefore, $v_i(A_i) \geq v'_i(A_i') \geq \alpha\cdot \mu_i$.
\end{proof}

\termination*
\begin{proof}
This is proof by contradiction. Suppose the algorithm stops at round $t$ because there are not enough items in $L$ ($=M\setminus J$) to satisfy any remaining agent $i$, i.e., $v_i(B_t \cup L)< \FRAC$. 

If $i\in N^1$, each removed bundle in rounds $k \in [t-1]$, has value of at most $1+\tfrac{3\gamma}{2}$ for agent $i$. Because, if $v_i(B_k) \ge \FRAC$ for $k \in [t-1]$, no more item has been added to $T= B_k$. Also, if $v_i(B_k) < \FRAC$ for $k \in [t-1]$ before adding the last item (if any) to $T$ the value of $T$ is less that $\FRAC$ and from Corollary \ref{cor:f-Bound-on-Values}, $v_{ij} < \tfrac{1}{4}+\tfrac{\gamma}{3}$ for $j \in L$. Therefore, at the end of the round $v_i(T) < 1+\tfrac{3\gamma}{2}$. 

This implies that the total value of assigned bundles in round 1 to $t-1$ is at most $(t-1)(1+\tfrac{3\gamma}{2})$. Let $n'$ be the number of agents remaining after Algorithm~\ref{algo:initial}, and $n$ be the number of agents in the original instance. We have $v_i(M) \ge n' - \frac{4\gamma r}{3}$ due to Lemma~\ref{lem:1/9n-MMS-Decrease}. Therefore, the value of items in $L$ before the round $t$ starts is at least $$v_i(L) \ge n' - \frac{4\gamma r}{3} - \left((t-1)(1+ \tfrac{3\gamma}{2}) + v_i(B_t) + (n'-t)(1+\tfrac{3\gamma}{2}) \right) . $$ This implies that $$v_i(L) + v_i(B_t) \ge 1- \frac{3(r+(n'-1))\gamma}{2} \ge 1 - \frac{3(n-1)\gamma}{2} \ge \frac{7}{8}+ \frac{1}{8n} \ge \frac{3}{4} + \gamma,$$ where we use $\gamma =\frac{1}{12n}$, which is a contradiction. 

If $i\in N^2$, then since at round $t$, $v_i(T) < \FRAC$, we have $B_t \in L_i$. Consider a round $k \in [t-1]$. If $B_k\not\in L_i$, then $T= B_k$ has been assigned to someone with no additional items added to $T$ from $L$ because $i \in \Gamma(T)$. If $B_k \in L_i$, then in round $k$, before adding the last item (if any) to $T$, the value of $i$ for $T$ is less that $\FRAC$ and from Lemma \ref{lem:f-three}, all items in $L$ have value of at most $\tf18$. Therefore, if $B_k \in L_i$, the value of the assigned bag for $i$ in round $k$ is less that $\tfrac{7}{8}+\gamma$. Hence, in the beginning of the round $t$, 
$$v_i(L) \ge \left(x_i+ \tfrac{l_i}{8} - \tf18 \right) - \left(x_i- (\tf34 +\gamma- v_i(B_t)) +(l_i-1)/8) \right) = \tf34+\gamma- v_i(B_t), $$ which is a contradiction.
\end{proof}

\giver*
\begin{proof}
If there exists a bundle $P_k$ with exactly one item $j$ from $J$, then $v_i(P_k \setminus J) = 1 - v_{ij} > \tfrac{1}{4}-\gamma$ because value of every item is less than $\FRAC$. Otherwise, each bundle has exactly two items from $J$, which implies that one of the bundles, say $P_{k}$, has two items $j_1,j_2$ from the set $\{n\} \cup J_2$. Since $v_{ij_1} + v_{ij_2} \leq v_{in} + v_{i(n+1)} <\FRAC$, $v_i(P_{k}\setminus J) >\tfrac{1}{4}-\gamma$.
\end{proof}

\fsp*
\begin{proof}
We show that if there are two items $j_1, j_2$ with each value more than $\tfrac{5}{8}+\gamma$ for an agent $i$ in a bundle $P_k$ of $P \in \MP^{n}_{i}(M\cup D)$, then we can construct another $P'\in \MP^n_{i}(M\cup D)$ where this is not true. Corollary~\ref{cor:f-Bound-on-Values} implies that there must exist another bundle $P_{k'}$ for which $\max_{j \in P_{k'}} v_{ij} < \tfrac{3}{8}+\tfrac{\gamma}{2}$. If there exists a partition $Q_{k'}^1$ and $Q_{k'}^2$ of items in $P_{k'}$ such that value of each $Q_{k'}^1$ and $Q_{k'}^2$ is at least $\tfrac{3}{8}-\gamma$, then we can rearrange items in $P_k \cup P_{k'}$ and make two new bundles $(\{j_1\} \cup Q_{k'}^1)$ and $(\{j_2\} \cup Q_{k'}^2)$. Clearly, the value of each bundle is at least $1$ and each has exactly one item with value more than $\tfrac{5}{8}+\gamma$.

If no such $Q_{k'}^1$ and $Q_{k'}^2$ exists, then we claim that there exists a partition of $P_{k'}$ into three sets with each value less than $\tfrac{3}{8}$. We can find this partition as follows: Initialize three empty bundles, and repeatedly add the highest value item of $P_{k'}$ to the bundle with the lowest value. For a contradiction, suppose one of the three bundles has value more than $\tfrac{3}{8}$, then the sum of the values of the other two sets must be less than $\tfrac{3}{8}$ because otherwise, they make a partition of two with each value more than $\tfrac{3}{8}-\gamma$. This means that at least one of the bundles must have a value less than $\tfrac{3}{16}$. This implies that the value of the largest bag before adding the last item must be less than$\tfrac{3}{16}$ and the last item also should value less than $\tfrac{3}{16}$, which is a contradiction. Therefore, there exists a partition $Q_{k'}^1$, $Q_{k'}^2$ and $Q_{k'}^3$ of $P_{k'}$ such that each value less than $\tfrac{3}{8}$.

According to Lemma \ref{lem:f-Giver-Bundle}, there exists a bundle $P_{\hat{k}} \in \MP^{n}_{i}(M)$ such that $v_i(P_{\hat{k}} \setminus J) > \tfrac{1}{4}-\gamma$. Let $v_i(P_{\hat{k}})=1+\delta$ for some $\delta \geq 0$. Observe that $P_{k'}$ cannot be same as $P_{\hat{k}}$ and $P_k$, otherwise we would have made two bundles earlier, each with value at least $1$ and exactly one item more than $\tfrac{5}{8}+\gamma$. We make three bags: $$\text{bag 1}: \{j_1\} \cup Q_{k'}^1 \hspace{2cm} \text{bag 2}: \{j_2\} \cup Q_{k'}^2 \hspace{2cm} \text{bag 3}: (P_{\hat{k}} \cap J) \cup Q_{k'}^3 $$

The value of each of bag 1 and bag 2 is at most $\tfrac{9}{8}+\tfrac{3\gamma}{2}$ and the value of bag 3 is at most $\tfrac{9}{8}+\gamma+\delta$. The total value of all items in $P_k$, $P_{k'}$ and $P_{\hat{k}}$ is at least $3+\tfrac{2}{8}+\delta+2\gamma$. We sort the remaining items in decreasing order and add them one by one to a bag with the lowest value. Since the value of the last item added to each bag has a value of at most $\tf18$ (from Lemma \ref{lem:f-three}), each bag has a value of at least 1 and it has at most one item with value more than $\tfrac{5}{8}+\gamma$. We repeat this process for all bundles with two item of value more than $\tfrac{5}{8}+\gamma$ to find a desired partition. 
\end{proof}

\fvmax*
\begin{proof}
For the first part, since $t_1=0$, we have $v_a(B_k) \geq 2y, \forall B_k\in L_a$ (Lemma~\ref{lem:f-y}). This implies that $x_a = (\FRAC)l_a - \sum_{k: B_k \in L_a} B_k \le (\FRAC - 2y)l_a$. 

For the second part, in $\bigcup_{k:B_k \in H_a} B_k $ there are exactly $2h_a$ items from $J$. If there are at least $2h_a$ items in $\bigcup_{t=t_1+1}^{t_2} P_t\cap J$ then $z=0$. On the other hand, if there are $2h_a-z$ items in $\bigcup_{t=t_1+1}^{t_2} P_t \cap J$, then it means that there are at least $z$ bundles in $P_{t_1+1},\dots, P_{t_2} $ with exactly one item from $J$. Each of these bundles need more than $\tfrac{1}{4}-\gamma$ value of items from $M\setminus J$ to become 1. 

Next, if $z=0$, then all bundles in $\{P_{t_2+1}, \dots, P_{t_3}\}$ need at least $x_a+l_a(\tfrac{1}{4}-\gamma)$ value of items from $M\setminus J$ to become $1$ because the value of each item in the bundles of $L_a$ is at least $y$ (follows from Definition~\ref{def:y} and the construction of $B_k$'s in \eqref{eq:B_i}). If $z>0$, then all bundles in $\{P_{t_2+1}, \dots, P_{t_3}\}$ need at least $x_a+l_a(\tfrac{1}{4}-\gamma) -yz$ value of items from $M\setminus J$ to become one because each of the $z$ items has value at most $y$ (Lemma~\ref{lem:f-y}). Therefore, in total, there should be at least $x_a+l_a(\tfrac{1}{4}-\gamma)- y z+z(\tfrac{1}{4}-\gamma)$ value of items from $M\setminus J$ in $\{P_{t_1+1}, \dots, P_{t_3}\}$. 

Further, in case $z$ is large such that the value of $x_a+l_a(\tfrac{1}{4}-\gamma)- y z+z(\tfrac{1}{4}-\gamma)$ is negative, we still need $z(\tfrac{1}{4}-\gamma)$ value of items from $M\setminus J$ to make all bundles in $\{P_{t_1+1},\dots, P_{t_2}\}$ value 1. 
\end{proof}

\maxva*
\begin{proof}
Using Lemma~\ref{lem:f-bound-vm-max}(2) there are two cases for the maximum of $x_a+l_a(\tfrac{1}{4}-\gamma)- y z+z(\tfrac{1}{4}-\gamma)$ and $z(\tfrac{1}{4}-\gamma)$. We prove the claim for each case separately.
\medskip

\noindent\textbf{Case 1:} Suppose $x_a+l_a(\tfrac{1}{4}-\gamma)- y z+z(\tfrac{1}{4}-\gamma) \ge z(\tfrac{1}{4}-\gamma)$, then 
        $z \leq \frac{x_a+ l_a(\tfrac{1}{4}-\gamma)}{y}$. Further, Lemma~\ref{lem:f-bound-vm-max}(1) implies $y\leq \frac{l_a(\FRAC)-x_a}{2l_a}$. 
   Using these, we get 
   \begin{align*}
       v_a(M \cup D \setminus J) \quad &\geq \quad x_a+l_a(\tfrac{1}{4}-\gamma) - y z+{z}(\tfrac{1}{4}-\gamma)\\
       &\ge \quad \frac{x_a(\tfrac{1}{4}-\gamma)+l_a(\tfrac{1}{4}-\gamma)^2}{y} \\
       & \ge \quad \frac{2x_al_a(\tfrac{1}{4}-\gamma)+2l_a^2(\tfrac{1}{4}-\gamma)^2}{l_a(\FRAC)-x_a} \enspace .
   \end{align*}
    \[  \]
We need to show that 
\begin{equation}\label{eqn:mt}
\frac{2x_al_a(\tfrac{1}{4}-\gamma)+2l_a^2(\tfrac{1}{4}-\gamma)^2}{l_a(\FRAC)-x_a} - \left( x_a+ \frac{l_a}{8}\right) \ge 0 \enspace . 
\end{equation}

For $n\le 4$ (i.e., number of agents in the original instance is at most 4), there is a simpler way to prove this claim. However, there are better approximation factors available for them in any case~\cite{ghodsi2017fair}, so we assume that $n>4$. Since $\gamma = \frac{1}{12n}$, it decreases as $n$ increases. Therefore, it is enough to show \eqref{eqn:mt} for $n=5$. For this, we put $\gamma=\frac{1}{60}$ in \eqref{eqn:mt} and get, 

\begin{align*}
\frac{2x_al_a(\tfrac{1}{4}-\tfrac{1}{60})+2l_a^2(\tfrac{1}{4}-\tfrac{1}{60})^2}{l_a(\tf34 + \tfrac{1}{60})-x_a} - \left( x_a+ \frac{l_a}{8}\right) \quad & = \quad \frac{7(30x_a l_a + 7l_a^2)}{15(23l_a - 30x_a)} - \frac{8x_a + l_a}{8} \\
& = \quad \frac{3600x_a^2 - 630x_a l_a + 47l_a^2}{120(23l_a - 30x_a)} \\
& > \quad  \frac{36(10x_a - l_a)^2 + l_a(11l_a + 90x_a)}{120(23l_a - 30x_a)}  \\
& \ge \quad 0 \enspace . 
\end{align*}

The last inequality follows because the denominator is positive using the definition of $x_a$. 
\medskip

\noindent\textbf{Case 2:} Suppose $x_a+l_a(\tfrac{1}{4}-\gamma) - y z+{z}(\tfrac{1}{4}-\gamma) \le {z}(\tfrac{1}{4}-\gamma)$. Then, using Lemma~\ref{lem:f-bound-vm-max}(1), we have
\[v_a((M \cup D) \setminus J) \geq {z}(\tfrac{1}{4}-\gamma) \geq \frac{x_a(\tfrac{1}{4}-\gamma)+l_a(\tfrac{1}{4}-\gamma)^2}{y} \geq x_a + \tfrac{l_a}{8}\enspace ,\] where the last inequality follows from the Case 1. 
\end{proof}

\thmfb*
\begin{proof}
Let $J_{\min} = \bigcup_{B_k\in L^1_a} B_k \cap J_2$. Then, we have $v_{aj}< \tf18, \forall  j\in J_{\min}$ because these items are bagged with an item value more than $\tfrac{5}{8}+\gamma$ and together they have value less than $\FRAC$. Therefore, $v_a(J_{\min}) < \tfrac{t_1}{8}$. Let $x_a' := (\FRAC)t_1  - \sum_{B_k\in L^1_a} v_a(B_k)$, and $x_a'' := x_a-x_a'$. For agent $a$, we can treat items of $J_{\min}$ as low-value items so we will prove:
 \begin{equation}
     \label{eq-extra-residual}
     v_a((M\cup D) \setminus J)+ v_a(J_{\min}) \geq x_a'+\tfrac{t_1}{4}+x''_a+{(l_a-t_1)}/{8} \text{.} 
 \end{equation}
 Since $v_a(J_{\min})$ is at most $\tfrac{t_1}{8}$, \eqref{eq-extra-residual} directly implies the theorem.
Let us call items in $(M\setminus J)\cup J_{\min} \cup D$ as \textit{filler} items (each has value $<\tf18$) and items in $J \setminus J_{\min}$ as \textit{base} items (each has value $> \tfrac{1}{4}+\tfrac{\gamma}{3}$).

In re-enumerated $\MP^{n}_{a}(M\cup D)$, recall that each of $\{P_1, \dots, P_{t_1}\}$ contains a base item whose value is at least the maximum of all the items in the remaining bundles. Further, if there are two base items in any bundle $P_k$ in $\{P_1, \dots, P_{t_1}\}$, then $v_a(P_k) > 1$. If each bundle in $\{P_1,\dots, P_{t_1}\}$ has exactly one base item ($t_1$ bundles in total) then items of value at least $x'_a+t_1(\tfrac{1}{4}-\gamma)$ are needed to make all these bundles at least 1 and there are $t_1$ items in $D_3$ with value $\gamma$. This, together with Theorem \ref{thm:maxva}, proves the bound. On the other hand, if there are more than $t_1$ \textit{base} items in $\{P_1,\dots, P_{t_1}\}$, then for each extra base item $j$, there are two cases. If $j$ comes from $B_k\in L^2_a$ then clearly, $x'_a$ will decrease by less than $v_{aj}$ but $x''_a$ will increase by $v_{aj}$, so the bound only improves. For the other case, if $j$ comes from $B_k\in H_a$ then it will make $B_k\setminus \{j\}$ to need more low-value items than the bundle in $\{P_1,\dots, P_{t_1}\}$ who gets $j$ to become one, so $x'_a$ will only increase and $x''_a$ stays same, and hence the bound only improves. This completes the proof.  
\end{proof}

\end{document}